\documentclass{article}
\usepackage[utf8]{inputenc}

\usepackage{xcolor}
\usepackage{graphicx}

\usepackage{amsmath}
\usepackage{amsthm}
\usepackage{amssymb}

\usepackage[linesnumbered,ruled,vlined]{algorithm2e}

\usepackage{cite}

\usepackage{fullpage}

\usepackage[linesnumbered,ruled,vlined]{algorithm2e}

\newtheorem{theorem}{Theorem} %[section]
\newtheorem{lemma}{Lemma}

\newtheorem{corollary}{Corollary}
\newtheorem{claim}{Claim}
\newtheorem{definition}{Definition}

\newtheorem{fact}{Fact}
\newtheorem{proposition}{Proposition}
\newtheorem{oq}{Open Question}

\newenvironment{lemma-repeat}[1]{\begin{trivlist}
\item[\hspace{\labelsep}{\bf\noindent Lemma \ref{#1} }]\em }%
{\end{trivlist}}
\newenvironment{theorem-repeat}[1]{\begin{trivlist}
\item[\hspace{\labelsep}{\bf\noindent Theorem \ref{#1} }]\em }%
{\end{trivlist}}

\newcommand{\remove}[1]{}

\DeclareMathOperator{\MIS}{MIS}
\DeclareMathOperator{\MaxIS}{MaxIS}

\def \poly {\text{poly}}

\newcommand{\ORL}{\textsc{Boppanna}}
\newcommand{\SORL}{\textsc{SeqBoppanna}}
\newcommand{\AIS}{\mathcal A}
\newcommand{\LMIS}{\textsc{RandMIS}}
\date{}
\begin{document}

\newcommand{\ThmMain}
{
	
		Given a weighted graph $G_w=(V,E,w)$, there is a constant $c>1$ and a $\poly(\log\log n)$-round algorithm in the CONGEST model that finds, with high probability, an independent set of weight at least $\frac{w(V)}{c\Delta}$.
	
}

\newcommand{\ThmMainAppEx}
{
	
	There is a randomized $(\poly(\log\log n)/\epsilon)$-round algorithm in the CONGEST model that finds, with high probability, a $(1+\epsilon)\Delta$-approximation for maximum-weight independent set.
	
}

\newcommand{\ThmMainD}
{
	
		Given a weighted graph $G_w=(V,E,w)$, there is an $O(\MIS(n,\Delta))$-round algorithm that finds an independent set of weight at least $\frac{w(V)}{4(\Delta+1)}$, in the CONGEST model. Whether the algorithm is deterministic or randomized depends on the $\MIS$ algorithm that is used as a black-box.
	
}

\newcommand{\ThmMainDApp}
{
	
	There is an $O(\MIS(n,\Delta)/\epsilon)$-round algorithm in the CONGEST model that finds a $(1+\epsilon)\Delta$-approximation for maximum-weight independent set. Whether the algorithm is deterministic or randomized, depends on the $\MIS$ algorithm that is run as a black-box.
	
}

\newcommand{\ThmMainArb}
{
	
	There is a randomized $O(\log n\cdot \poly\log\log n/\epsilon)$-round algorithm in the CONGEST model that finds, with high probability, an $8(1+\epsilon)\alpha$-approximation for maximum-weight independent set.
	
}

\begin{titlepage}
	\title{Improved Distributed Approximations for Maximum-Weight Independent Set}
	\author{Ken-ichi Kawarabayashi\thanks{NII, Japan, \texttt{kkeniti@nii.ac.jp}}
		\and Seri Khoury\thanks{University of California, Berkeley,  \text{seri\_khoury@berkeley.edu}}\and Aaron Schild\thanks{University of California, Berkeley, \texttt{aschild@berkeley.edu}}\and Gregory Schwartzman\thanks{NII, Japan, \texttt{greg@nii.ac.jp}}}
	
		\maketitle

	\begin{abstract}
		We present improved results for approximating maximum-weight independent set ($\MaxIS$) in the CONGEST and LOCAL models of distributed computing. Given an input graph, let $n$ and $\Delta$ be the number of nodes and maximum degree, respectively, and let $\MIS(n,\Delta)$ be the the running time of finding a \emph{maximal} independent  set ($\MIS$) in the CONGEST model. Bar-Yehuda et al. [PODC 2017] showed that there is an algorithm in the CONGEST model that finds a $\Delta$-approximation for $\MaxIS$ in $O(\MIS(n,\Delta)\log W)$ rounds, where $W$ is the maximum weight of a node in the graph, which can be as high as $\poly (n)$. Whether their algorithm is deterministic or randomized depends on the $\MIS$ algorithm that is used as a black-box. Our results:
		\begin{enumerate}
			\item  A deterministic $O(\MIS(n,\Delta)/\epsilon)$-round algorithm that finds a $(1+\epsilon)\Delta$-approximation for $\MaxIS$ in the CONGEST model.
			\item A randomized $(\poly(\log\log n)/\epsilon)$-round algorithm that finds, with high probability, a $(1+\epsilon)\Delta$-approximation for $\MaxIS$ in the CONGEST model. That is, by sacrificing only a tiny fraction of the approximation guarantee, we achieve an \emph{exponential} speed-up in the running time over the previous best known result. 		Due to a lower bound of $\Omega(\sqrt{\log n/\log \log n})$ that was given by Kuhn, Moscibroda and Wattenhofer [JACM, 2016] on the number of rounds for any (possibly randomized) algorithm that finds a maximal independent set (even in the LOCAL model) this result implies that finding a $(1+\epsilon)\Delta$-approximation for $\MaxIS$ is exponentially easier than $\MIS$. 
			\item A randomized $O(\log n\cdot \poly(\log\log n)/\epsilon)$-round algorithm that finds, with high probability, a $8(1+\epsilon)\alpha$-approximation for $\MaxIS$ in the CONGEST model, where $\alpha$ is the arboricity of the graph. For graphs of arboricity $\alpha<\Delta/(8(1+\epsilon))$, this result improves upon the previous best known result in both the approximation factor and the running time. 
		\end{enumerate}
		
		One may wonder whether it is possible to approximate $\MaxIS$ in fewer than $\poly(\log\log n)$ rounds. We show that this is possible for unweighted graphs of maximum degree $\Delta\leq n/\log n$. For such graphs, we give a randomized $O(1/\epsilon)$-round algorithm in the CONGEST model that, with high probability, finds an independent set of size at least $\frac{n}{(1+\epsilon)(\Delta+1)}$, which is a $(1+\epsilon)(\Delta+1)$-approximation to the optimal solution. This result cannot be extended to very high degree graphs, as we show a lower bound of $\Omega(\log^*n)$ rounds for any (possibly randomized) algorithm that finds an independent set of size $\Omega(n/\Delta)$ in unweighted graphs, even in the LOCAL model. The hard instances that we use to prove our lower bound are graphs of maximum degree $\Delta=\Omega(n/\log^*n)$.
	\end{abstract}
\thispagestyle{empty}
\end{titlepage}

%\title{Improved Algorithms for Maximum Independent Set Approximation for CONGEST, CONGESTED CLIQUE, and Massively Parallel Computation}

\newpage

\section{Introduction and Related Work}

One of the most fundamental problems in distributed graph algorithms is the \emph{maximal independent set} problem ($\MIS$), where given an input graph, we need to find a maximal subset of the nodes such that no two nodes in the subset are adjacent. This problem has received a tremendous amount of attention in various distributed models (see for example~\cite{Ghaffari19,Ghaffari2016,GhaffariGKMR18,BarenboimEK14,BarenboimEPS12,FirstBarenboimEPS16,AlonBI86,luby1986simple,LenzenW11,KuhnMW16,RozhonGhaffari,Balliu0HORS19,AwerbuchGLP89,PanconesiR01,PanconesiS96,linial1992locality,Naor91}). It is considered one of the four classic problems of local distributed
algorithms, along with edge coloring, vertex coloring, and maximal matching~\cite{2013Barenboim,PanconesiR01,FischerGK17}.

Independent sets have many applications in practical and theoretical computer science. Especially independent sets of large size. These include applications in economics~\cite{bodino1962economic}, computational biology~\cite{Butenko05clique-detectionmodels,Barticle}, coding theory~\cite{Bomze99themaximum,ButenkoPSSS02}, and experimental design~\cite{BalasY86}. In unweighted graphs, a \emph{maximum} independent set is an independent set of maximum size. In weighted graphs, a \emph{maximum-weight} independent set ($\MaxIS$) is an independent set of maximum total weight, where by total we mean the sum of weights of nodes in the independent set.

%The classic algorithm by~\cite{} show that $\MIS$ can be solved in $O(\log n)$ rounds in the standard LOCAL model. In recent breakthroughs, Ghaffari showed faster algorithms on graphs of small maximum degree (sub-polynomial in $n$)\cite{}. 
In an unweighted graph, any $\MIS$ constitutes a $\Delta$-approximation for $\MaxIS$, where $\Delta$ is the maximum degree of a node in the graph. This implies that $\MIS$ cannot be easier than $\Delta$-approximation for $\MaxIS$ in unweighted graphs, regardless of the computational model. This leaves a natural question of whether $\Delta$-approximation for $\MaxIS$ is easier than $\MIS$. In the classical sequential setting, finding an $\MIS$ has the same complexity as finding a $\Delta$-approximation for $\MaxIS$ (even in weighted graphs) as both problems admit simple linear-time greedy algorithms. In this work, we show that in the distributed setting, finding a $(1+\epsilon)\Delta$-approximation for $\MaxIS$ is exponentially easier than $\MIS$.

\subsection*{Distributed Computing and Our Results}

The major two models of distributed graph algorithms are the LOCAL and CONGEST models. In the LOCAL model~\cite{linial1992locality}, there is a synchronized communication network of $n$ computationally-unbounded nodes, where each node has a unique $O(\log n)$-bit identifier. In each communication round, each node can send an unbounded-size message to each of its neighbors. The task of the nodes is to compute some function of the network (e.g., its diameter, the value of a maximum independent set, etc.), while minimizing the number of communication rounds. The CONGEST model~\cite{Peleg:2000} is similar to the LOCAL model, where the only difference is that the message-size is bounded by $O(\log n)$ bits.

In this work we study the problem of approximating $\MaxIS$, which has been studied in both the LOCAL and CONGEST models~\cite{GhaffariKM17,lenzen2008leveraging, czygrinow2008fast,Censor-HillelKP17,bodlaenderbrief,HalldorssonK18,BachrachCDELP19}. In unweighted graphs, one can find a $\Delta$-approximation for $\MaxIS$ by finding an $\MIS$. In recent years, our understanding of the complexity of $\MIS$ has been substantially improving\cite{Ghaffari19,Ghaffari2016,FirstBarenboimEPS16,RozhonGhaffari}, leading to a recent remarkable breakthrough by Rozhon and Ghaffari~\cite{RozhonGhaffari}, where they show a deterministic $\poly(\log n)$-round algorithm for finding an $\MIS$, even in the CONGEST model. This result also implies a randomized algorithm that finds an $\MIS$ with high probability in $O(\log\Delta)+\poly(\log\log n)$ rounds, in the CONGEST model\footnote{We say that an algorithm succeeds with high probability if it succeeds with probability $1-1/n^c$ for an arbitrary constant $c>1$.}~\cite{RozhonGhaffari,Ghaffari19,GhaffariPersonal,Censor-HillelPS17}. 
 %The lower bound of~\cite{KuhnMW16} holds against randomized algorithms that succeed with high probability as well.

In a weighted graph, an $\MIS$ doesn't necessarily constitute a $\Delta$-approximation for $\MaxIS$. For the weighted case, Bar-Yehuda et al.~\cite{Bar-YehudaCGS17} showed a $\Delta$-approximation algorithm in the CONGEST model that takes $O(\MIS(n,\Delta)\cdot \log W)$ rounds, where $\MIS(n,\Delta)$ is the running time for finding an $\MIS$ in graphs with $n$ nodes and maximum degree $\Delta$,  and $W$ is the maximum weight of a node in the graph (which can be as high as $\poly(n)$). Whether their algorithm is deterministic or randomized, depends on the $\MIS$ algorithm that is used as a black-box.

In this work we present faster algorithms compared to~\cite{Bar-YehudaCGS17}, by paying only a $(1+\epsilon)$ multiplicative overhead in the approximation factor. Our main result (Theorem~\ref{thm:Wmain}) is a randomized algorithm that achieves an exponential speed-up compared to~\cite{Bar-YehudaCGS17}. Our results:

\begin{theorem}\label{thm:main}

	\ThmMainDApp
\end{theorem}

\begin{theorem}\label{thm:Wmain}
	\ThmMainAppEx
\end{theorem}

Due to a lower bound of $\Omega(\sqrt{\log n/\log \log n})$ that was given by Kuhn, Moscibroda and Wattenhofer~\cite{KuhnMW16}, against any (possibly randomized) algorithm that finds an $\MIS$, even in the LOCAL model, Theorem~\ref{thm:Wmain} implies that finding a $(1+\epsilon)\Delta$-approximation for $\MaxIS$ is exponentially easier than $\MIS$.  

Using the algorithm from Theorem~\ref{thm:Wmain}, we can also get an improved approximation algorithm for a wide range of arboricity. Let $\alpha$ be the arboricity of the input graph (See also Definition~\ref{def:arb}). For graphs of arboricity $\alpha\leq \Delta/(8(1+\epsilon))$, Theorem~\ref{thm:mainArb} improves upon~\cite{Bar-YehudaCGS17} in both the running time and approximation factor.

\begin{theorem}\label{thm:mainArb}
	\ThmMainArb
\end{theorem}

\paragraph{Results for unweighted graphs:} Recently, Boppana et al.~\cite{BoppanaHR18} showed that running the one-round classical ranking algorithm yields a solution with an \emph{expected} weight at least $w(V)/(\Delta+1)$, where $w(V)$ is the total weight of nodes in the graph.\footnote{To the best of our knowledge, the classical ranking algorithm has first appeared in the book of Alon and Spencer~\cite{AlonS92} and is due to Boppana (see also the references for this algorithm in~\cite{BoppanaHR18}).} In the classical ranking algorithm, each node $v$ picks a number $r_v$ uniformly at random in $[0,1]$. If $r_v>r_u$ for any neighbor $u$ of $v$, then $v$ joins the independent set. Since every node joins the independent set with probability at least $1/(\Delta+1)$, the expected weight of the independent set is at least $w(V)/(\Delta+1)$. However, algorithms that work well in expectation don't necessarily work well with good probability. In fact, for the algorithm given by~\cite{BoppanaHR18}, it is not very hard to construct examples in which the \emph{variance} of the solution is very high, in which case the algorithm doesn't return the expected value with high probability. In this work we prove the following stronger theorem for any algorithm.

\begin{theorem}\label{thm:LB}
	Any algorithm that finds an independent set of size $\Omega(n/\Delta)$ in unweighted graphs, with success probability $p\geq 1-1/\log n$ must spend $\Omega(\log^*n)$ rounds, even in the LOCAL model.
\end{theorem}

Interestingly, this hardness result applies for graphs of maximum degree $\Delta=\Omega(n/\log^*n)$. One may wonder whether we can extend the lower bound for much smaller maximum degree graphs. We rule out this possibility, with the following theorem. The proof of Theorem~\ref{thm:LD} relies on a novel idea for analyzing the classical ranking algorithm using \emph{martingales}, and the \emph{local-ratio} technique, on which we elaborate in the technical overview. 

\begin{theorem}\label{thm:LD}
	For unweighted graphs of maximum degree $\Delta\leq n/\log n$, there is an $O(1/\epsilon)$-round algorithm in the CONGEST model that finds, with high probability, an independent set of size at least  $\frac{n}{(1+\epsilon)(\Delta+1)}$.
\end{theorem}

\paragraph{Further Related Work.}\label{RW} 

Ghaffari et al.~\cite{GhaffariKM17}, showed that there is an algorithm for the LOCAL model that finds a $(1+\epsilon)$-approximation for $\MaxIS$ in $O(\poly(\log n/\epsilon))$ rounds, for a constant $\epsilon$. The results in \cite{lenzen2008leveraging, czygrinow2008fast} give a lower bound of $\Omega(\log^* n)$ rounds for any deterministic algorithm that returns an independent set of size at least $n/\log^*n$ on a cycle, and a randomized $O(1)$-round algorithm for $O(1)$-approximations in planar graphs, in the LOCAL model. The results by~\cite{bodlaenderbrief,HalldorssonK18} give fast algorithms for approximating $\MaxIS$ in unweighted graphs, where the approximation guarantees are only in expectation. 

\paragraph{Road-map:} 
In Section~\ref{sec:tech} we provide a technical overview. Section~\ref{sec:Prelim} contains some basic definitions and useful inequalities. The technical heart of the paper starts in Section~\ref{sec:C}, where we prove our first two results (Theorems~\ref{thm:main} and~\ref{thm:Wmain}). Our results for low-degree and low-arboricity graphs are presented in Sections~\ref{sec:low-deg} and~\ref{sec:arb}, respectively. Our lower bound result is presented in Section~\ref{sec:LB}. Finally, we conclude the paper with a discussion and open questions in Section~\ref{sec:dis}.

\section{Technical Overview}\label{sec:tech}

\paragraph{Results for weighted graphs:} Our first two results (Theorems~\ref{thm:main} and~\ref{thm:Wmain}) share a similar proof structure. First, we show that there are fast algorithms for $O(\Delta)$-approximation. Then we use the \emph{local-ratio} technique~\cite{bar2001unified} to prove a general boosting theorem that takes a $T$-round algorithm for $O(\Delta)$-approximation, and use it as a black-box to output a $(1+\epsilon)\Delta$-approximation in $O(T/\epsilon)$ rounds. An overview of the local-ratio technique and the boosting theorem is provided in Section~\ref{sec:amp}. The key ingredient to show a fast $O(\Delta)$-approximation algorithm is a new \emph{weighted sparsification} technique, where we show that it suffices to find an independent set of a good approximation in a sparse subgraph. An overview of the weighted sparsification technique is provided in Section~\ref{sec:O(D)}.

Our improved approximation algorithm for low-arboricity graphs (Theorem~\ref{thm:mainArb}) uses Theorem~\ref{thm:Wmain} as a black-box, where the main technical ingredient is the local-ratio technique. An overview of this algorithm is also provided in Section~\ref{sec:amp}.

\paragraph{Results for unweighted graphs:} Our upper bound for unweighted graphs of maximum degree $\Delta\leq n/\log n$ (Theorem~\ref{thm:LD}) has a similar two-step structure as the first two results. We first show an $O(\Delta)$-approximation algorithm, and then we use the local-ratio technique to boost the approximation factor. For the $O(\Delta)$-approximation part, we show that running the classical one-round ranking algorithm (that was used by~\cite{BoppanaHR18}) for $c$ rounds already returns an $O(\Delta)$-approximation for unweighted graphs of maximum degree $\Delta\leq n/\log n$, with probability $\approx 1-1/n^c$. The main technical ingredient for showing this result is a new analysis of the classical ranking algorithm using \emph{martingales}. An overview of this result is provided in Section~\ref{sec:EM}. Finally, in Section~\ref{sec:TLB}, we provide an overview of the lower bound result (Theorem~\ref{thm:LB}).  
\subsection{Weighted Sparsification for $\boldsymbol{O(\Delta)}$-Approximation}\label{sec:O(D)}

A good way to understand the $O(\Delta)$-approximation algorithm is to first consider the unweighted case. Let $G=(V,E)$ be an unweighted graph. We can find an $O(\Delta)$-approximation for $\MaxIS$ in $G$ as follows. First, we sample a sparse subgraph $H$ of $G$ with the following properties. (1) The maximum degree $\Delta_H$ of $H$ is small ($O(\log n)$). (2) The ratio between the number of nodes ($n_H$) and the maximum degree of $H$ is at least as in $G$, up to a constant multiplicative factor. That is, $n_H/\Delta_H=\Omega(n/\Delta)$. Since any $\MIS$ in $H$ has size at least $n_H/\Delta_H=\Omega( n/\Delta)$, it suffices to find an $\MIS$ in $H$, which take $\MIS(n_H,\Delta_H)\leq \MIS(n,\log n)$ rounds (recall that $\MIS(n,\Delta)$ is the running time of finding an $\MIS$ in graphs of $n$ nodes and maximum degree $\Delta$). By the recent breakthrough of Rozhon and Ghaffari~\cite{RozhonGhaffari}, $\MIS(n,\log n)=O(\log\log n) +\poly(\log\log n)=\poly\log\log n$ rounds. Furthermore, sampling a subgraph with the aforementioned properties is almost trivial. Each node joins $H$ with probability $\min\{\log n/\Delta,1\}$, independently. It is not very hard to show, via standard Chernoff (Fact~\ref{fact:chrnoff}) and Union Bound arguments, that $H$ has the desired properties. While this approach is straightforward for the unweighted case, it runs into challenges when trying to apply it for the weighted case, as we explain next.

\paragraph{The challenge in weighted graphs:} Perhaps the first thing that comes into mind when trying to extend the sampling technique to weighted graphs is to try to sample a sparse subgraph $H$ with the following properties. (1) The maximum degree $\Delta_H=O(\log n)$. (2) The ratio between the \emph{total weight} in $H$ and the max degree of $H$ is the same as in $G$, up to a constant multiplicative factor. That is $w(V_H)/\Delta_H=\Omega( w(V)/\Delta)$, where $w(V_H)$ is the total weight of nodes in $H$ and $w(V)$ is the total weight of nodes in $G$. However, this approach runs into two challenges. The first challenge is that in the weighted case, an $\MIS$ doesn't necessarily constitute a $\Delta$-approximation for $\MaxIS$. Therefore, even if we are able to sample a subgraph $H$ with the desired properties, running an $\MIS$ algorithm on $H$ might result in an independent set of a very small weight. To overcome this challenge, we show a very simple $\MIS(n,\Delta)$-round algorithm that finds an $O(\Delta)$-approximation. This algorithm runs an $\MIS$ algorithm on the subgraph induced by nodes that are relatively heavy, compared to their neighbors. Specifically, a node is considered relatively heavy compared to its neighbors, if it is of weight at least $\Omega(1/\Delta)$-fraction of the sum of weights of its neighbors. It is not very hard to show that this algorithm returns an independent set of total weight $\Omega(w(V)/\Delta)$, where $w(V)$ is the total weight of nodes in the graph. The proof of this argument is provided in Section~\ref{sec:warmup}.

Furthermore, another challenge is that the same sampling procedure doesn't work for the weighted case. In particular, if we sample each node with probability $p=\min\{(\log n)/\Delta,1\}$, then \emph{light}-weight nodes will have the same probability of joining $H$ as \emph{heavy}-weight nodes. Intuitively, we need to take the weights into account. For this, we boost the sampling probability of a node $v$ by an additive factor of $w(v)\log n/w(V)$, where $w(v)$ is the weight of $v$ and $w(V)$ is the total weight of nodes in the graph. In order to show that the sampled subgraph has the desired properties, it doesn't suffice to use standard Chernoff and Union-Bound arguments. Instead, we present a more involved analysis that uses Bernstein's inequality (Fact~\ref{fact:Berns}). Observe that the nodes don't know the value $w(V)$. Therefore, we define a notion of \emph{weighted degree} of a node, which is the sum of weights of its neighbors. We show that it suffices for a node $v$ to use the maximum weighted degree in its neighborhood, instead of $w(V)$. The full argument is provided in Section~\ref{sec:WeightedSparsifier}.

\subsection{Boosting the Approximation Factor using Local-Ratio}\label{sec:amp}

A useful technique for approximation algorithms is the local-ratio technique~\cite{bar2001unified}. In recent years, the local-ratio technique has been found to be very useful for the distributed setting~\cite{Bar-YehudaCGS17,Bar-YehudaCS17}, and the $\Delta$-approximation algorithm of~\cite{Bar-YehudaCGS17} also uses this technique. In this work we use local-ratio to boost the approximation guarantee for $\MaxIS$. We start with stating the local-ratio theorem for maximization problems. Here, we state it specifically for $\MaxIS$. Given a weighted graph $G_w=(V,E,w)$, where $w$ is a  node-weight function $w:V\rightarrow \mathbb{R}$, we say that an independent set $I\subseteq V$ is $r$-approximate with respect to $w$ if it is $r$-approximate for the optimal solution in $G_w$.

\begin{theorem}\label{thm:LR}\textbf{[Theorem 9 in~\cite{bar2001unified}]}\\
	Let $G_w=(V,E,w)$ be a weighted graph. Let $w_1$ and $w_2$ be two node-weight functions such that $w=w_1+w_2$. If an independent set $I$ is $r$-approximate with respect to $w_1$ and with respect to $w_2$ then it is $r$-approximate with respect to $w$ as well.
\end{theorem}

Theorem~\ref{thm:LR} already gives a simple linear-time sequential algorithm for $\Delta$-approximation for $\MaxIS$, as follows. Pick an arbitrary node $v$ of positive weight, push it onto a stack, and reduce the weight of any node in the inclusive neighborhood of $v$ ($v$ and its neighbors) by $w(v)$. Continue recursively on the obtained graph, until there are no nodes of positive weight. When there are no remaining nodes of positive weight, pop out the stack, and construct an independent set $I$ greedily, as follows. For each node $v$ that is popped out from the stack, add $v$ to $I$, unless it already contains a neighbor of $v$. 

The reason that this simple algorithm gives a $\Delta$-approximation is as follows. Consider the first iteration, when the algorithm picks an arbitrary node $v$, pushes it onto a stack, and reduces the weight of any node in the inclusive neighborhood of $v$ by $w(v)$. This first iteration implicitly defines two weight functions: the \emph{reduced} weight function $w_1$, and the \emph{residual} weight function $w_2$, where $w=w_1+w_2$. That is, the reduced weigh of a node $u$ in the first step is $w_1(u)=w(v)$ if it belongs to the inclusive neighborhood of $v$, and $w_1(u)=0$ otherwise. The residual weight of a node $u$ is the remaining weight $w_2(v)=w(v)-w_1(v)$. To prove that the algorithm returns a $\Delta$-approximation, we can assume by reverse induction that $I$ is a  $\Delta$-approximation with respect to the residual weight function $w_2$. Furthermore, the independent set is constructed in a way such that it must contain at least one node in the inclusive neighborhood of $v$, where the weight of this node with respect to $w_1$ is $w(v)$. Since the degree of $v$ is at most $\Delta$, and the value of the optimal solution with respect to $w_1$ is at most $\Delta w(v)$, it follows that $I$ is also $\Delta$-approximation with respect to the reduced weight function $w_1$. Hence, by the local-ratio theorem, the independent set is also a $\Delta$-approximation with respect to $w=w_1+w_2$.

One can extend this idea, and rather than picking a single node in each step, the algorithm can pick an arbitrary independent set $I'$, push all the nodes in $I'$ onto a stack, and perform local weight reductions in the inclusive neighborhood of any node in $I'$. The algorithm continues recursively on the obtained graph after the weight reductions, until there are no remaining nodes of positive weight. Then, the algorithm constructs an independent set $I$ by popping out the stack and adding nodes in the stack to $I$ greedily. Using a similar local-ratio argument, one can show that this algorithm also returns a $\Delta$-approximation for $\MaxIS$. The idea of picking an independent set rather than a single node in each step was used by~\cite{Bar-YehudaCGS17} to show a $\Delta$-approximation algorithm in $O(\MIS(n,\Delta)\log W)$ rounds.

In this work, we prove a simple yet powerful property about the local-ratio technique. Specifically, we show that the total weight of the independent set $I$ that is constructed in the pop-out stage (with respect to the original input weight function $w$), is at least the total weight of the nodes in the stack (with respect to the residual weight function at the time they were pushed onto the stack). That is, let $S$ be set of nodes that are pushed onto the stack. For $v\in S$, let $w_{i_v}$ be the residual weight of $v$ at the time it was pushed onto the stack. We prove (Proposition~\ref{prop:stack} in Section~\ref{sec:amplification}) that $w(I)\geq \sum_{v\in S} w_{i_v}(v)$. We refer to this property as the \emph{stack property}.

The stack property allows us to show a general boosting theorem, as follows. We use the local-ratio algorithm described above, where in each step we pick an independent set $I'$ that is $(c\Delta)$-approximation for $\MaxIS$, for some constant $c>1$. Hence, intuitively, after $\approx c/\epsilon$ steps, the total weight in the stack should be at least $\frac{OPT(G_w)}{(1+\epsilon)\Delta}$, where $OPT(G_w)$ is the value of an optimal solution in the input graph $G_w$. The full argument of the boosting theorem is provided in Section~\ref{sec:amplification}.

\paragraph{Low-arboricity graphs:} Moreover, the stack property allows us to show an improved approximation algorithm for low-arboricity graphs, as follows. In each step, we run a $(1+\epsilon)\Delta$-approximation algorithm on the subgraph induced by the nodes of degree at most $4\alpha$, where $\alpha$ is the arboricity of the graph. We push the nodes in the resulting independent set $I'$ onto the stack, and perform local weight reduction in the neighborhoods of the nodes in $I'$. Then, we delete all the nodes of degree at most $4\alpha$, and continue recursively on the resulting graph. Finally, we construct an independent $I$ by popping out the stack greedily. By a standard Markov argument, after $\log n$ push steps, the graph becomes empty. Furthermore, since in each step the algorithm finds a $(1+\epsilon)4\alpha$ approximation in the subgraph induced by the nodes of degree at most $4\alpha$, and this independent set is pushed onto the stack, we are able to use the stack property to show that the constructed independent set $I$ is roughly of the same approximation for $G_w$. The full argument for low-arboricity graphs is provided in Section~\ref{sec:arb}.

\subsection{Analysis of the Ranking Algorithm using Martingales}\label{sec:EM}

In this section we provide an overview of our result for unweighted graphs of maximum degree $\Delta\leq n/\log n$  (Theorem~\ref{thm:LD}).  First, we find an $O(\Delta)$-approximation, and then we use the boosting theorem to get a $(1+\epsilon)\Delta$-approximation. To find an $O(\Delta)$-approximation, we use the classical ranking algorithm. Recall that in the ranking algorithm, each node $v$ picks a number $r_v$ uniformly at random in $[0,1]$. If $r_v>r_u$ for any neighbor $u$ of $v$, then $v$ joins the independent set. Let $I$ be the independent set that is returned by the ranking algorithm. The crux of the analysis is in using concentration inequalities to get a high-probability lower bound on the number of nodes in $I$. However, it is unclear how to make this approach work, as the random variables $X_v = \textbf{1}_{v\in I}$ are not independent. While these random variables are not independent, one can obtain a weaker result in this direction. Specifically, for graphs of maximum degree at most $n^{1/3}/\text{poly}(\log n)$, one can get a useful bound on the maximum \emph{dependency} among these variables. In particular, one can show that each $X_v$ is dependent on at most $(n^{1/3}/\text{poly}(\log n))^2 = n^{2/3}/\text{poly}(\log n)$ other $X_u$s, which makes it possible to show concentration using the bounded dependence Chernoff bound given in \cite{Pemmaraju01}. However, it is unclear how to use this approach for higher degree graphs.

The main idea of our approach is to view the ranking algorithm from a sequential perspective. Instead of picking ranks for the nodes and including a node in $I$ if its rank is higher than that of its neighbors, we draw nodes $v$ from $V$ uniformly at random one at a time and add $v$ to $I$ if it is not adjacent to any previously drawn node. We show that the resulting independent set is identical in distribution to the independent set produced by the ranking algorithm (Proposition~\ref{prop:indep-equiv} in Section~\ref{sec:low-deg}). Note that this is not the same as a sequential greedy algorithm for maximal independent set, which would add $v$ to $I$ if it is not adjacent to any node in $I$ (a weaker condition). The sequential perspective of the ranking algorithm allows us to think about the size of $I$ incrementally. One could directly show concentration if the family of random variables $\{I_t\}_t$ was a martingale. However, this is not the case, as $|I_{t+1}| \ge |I_t|$ so it is not possible for expected increments to be 0. Instead, we create a martingale by shifting the increments so that they have mean 0. More formally, let $I_t$ be the independent set $I$ after the first $t$ nodes have been drawn. Let $v_t$ be the $t$th node drawn. The random variable

$$Y_t = |I_t| - |I_{t-1}| - \Pr[v_t\in I | I_{t-1}]$$
has mean 0 conditioned on $I_{t-1}$. Therefore, the $Y_t$s are increments for the martingale $X_t = \sum_{i=1}^t Y_t$. Using Azuma's Inequality, one can show that $X_t$ concentrates around its mean, which is 0. To lower bound the size of the obtained independent set $I$, one therefore just needs to get a lower bound on the sum of the increment probabilities $\Pr[v_t\in I | I_{t-1}]$. This can be lower bounded by 1/2 when $t= o(n/\Delta)$ because when a node is drawn, it eliminates at most $\Delta$ other nodes from inclusion into $I$. But when $t = \Theta(n/\Delta)$, the sum of these probabilities is already $1/2(\Theta(n/\Delta)) = \Theta(n/\Delta)$, so the independent set is already large enough, as desired. The reason that this technique works for $\Delta\leq n/\log n$ is that the success probability is roughly exponential in $n/\Delta$. Hence, by having $\Delta\leq n/\log n$, we get a high probability success, as desired. The full argument is provided in Section~\ref{sec:low-deg}.

\subsection{An Overview of the Lower Bound}\label{sec:TLB}
In this section we give an overview of our lower bound (Theorem~\ref{thm:LB}). For the  deterministic case, one can show a lower bound for finding an independent set of size $\Omega(n/\Delta)$ in a cycle, by a reduction to the classical lower bound of Linial for finding an $\MIS$ in a cycle~\cite{linial1992locality}. However, for the randomized case, this approach becomes more challenging. In fact, the cycle graph cannot be a hard instance for finding an independent set of size $\Omega(n/\Delta)$, as  there is a constant-round algorithm for low degree graphs (as  we show in Theorem~\ref{thm:LD}). In order to show hardness for the randomized case, we use a \emph{cycle of cliques} graph. We are able to reduce the problem of finding an independent set of size $\Omega(n/\Delta)$ in a cycle of cliques, to the problem of finding an $\MIS$ in a cycle. And we use Naor's lower bound~\cite{Naor91} for finding an $\MIS$ in a cycle, which holds even against randomized algorithms. We start by stating Naor's lower bound.

\begin{theorem}\label{fact:mis-lb}(\textbf{Lower bound for the cycle} \cite{Naor91}).
	Any randomized algorithm in the LOCAL model for finding a  maximal independent set that takes fewer than $\frac{1}{2}(\log^* n) - 4$ rounds, succeeds with probability at most $1/2$, even for a cycle of length $n$.
\end{theorem}

Perhaps a good way to understand our reduction to Naor's lower bound is to first consider deterministic algorithms. Let $\AIS$ be a deterministic algorithm for approximate MaxIS. Suppose that it takes $T(n)$ rounds in graphs of $n$ nodes. We can use $\AIS$ to find a maximal independent set in a cycle $C$ of $n$ nodes, as follows. We start by running  $\AIS$ on $C$ to produce an independent set $I$. Since $C$ is a cycle, there is a natural clockwise ordering for the nodes of $I$. Between any two consecutive nodes of $I$, there may be nodes along the cycle that are not adjacent to a node in $I$. We informally call these nodes the ``gaps'' between consecutive nodes in $I$. We can obtain a maximal independent set in $C$ by ``filling in'' the gap between every two consecutive nodes in $I$ with a maximal independent set (sequentially). To bound the runtime of this algorithm, we need to bound the maximum length of a gap. Since $\AIS$ is deterministic, it is not very hard to show that the maximum length of a gap is $O(T(n))$. This is because from a local perspective, the nodes cannot distinguish between $C$ and a path of length $\omega(T(n))$, by a standard indistinguishability argument. Hence, one can show that if there is a gap of length  $\omega(T(n))$, then $\AIS$ doesn’t return the required approximation on a path of length $\omega(T(n))$. As a result, filling in the  gaps between nodes in $I$ takes $O(T(n))$ rounds. Therefore, by running $\mathcal{A}$ on $C$ and then filling in the gaps sequentially, we get an $\MIS$ in $O(T(n))$ rounds. And by Linial's lower bound~\cite{linial1992locality}, we have that $T(n)=\Omega(\log^*n)$.

However, the argument above fails if $\AIS$ is a randomized algorithm. The main issue is that when running a randomized algorithm on a cycle, the maximum length of a gap between two consecutive nodes in the independent set can be larger than $O(T(n))$. This is because randomized algorithms that succeed with high probability can fail with probability $1/\poly(n)$, where $n$ is the number of nodes in the graph. Hence, $\AIS$  can fail on a path of length $O(T(n))$ with probability $1/\poly(T(n))$ which is non-negligible when $T(n)\ll n$. In particular, since there are $\Omega(|C|/T(n))=\Omega(n/T(n))$ subpaths of length $O(T(n))$ in $C$, it is likely that $\mathcal{A}$ fails on at least one of these subpaths. If on the the other hand the number of nodes in the $O(T(n))$-radius neighborhood of a node was larger, then one could hope to get around this issue, as it would amplify the ``local" success probability in the neighborhood of a node.

Hence, instead of running $\mathcal{A}$ on $C$, we run it on a cycle of cliques $C_1$, which is obtained from $C$ as follows. Each node $v\in C$ is replaced with a clique of size $\approx 2^{|C|}$, denoted by $D(v)$, where every two adjacent cliques are connected by a bi-clique (see also Figure 1, for an illustration). By running $\AIS$ on $C_1$ instead of $C$, it boosts the success probability of $\AIS$ in a small-radius neighborhood of any given node. As a result, a small-radius neighborhood of any node in $C_1$ must contain a node in the independent set. Using the independent set $I_1$ that was found in $C_1$, we can map it to an independent set $I$ in $C$, as follows. Every $v\in C$ joins $I$ if and only if $I_1$ contains a node in $D(v)$. Due to the approximation guarantee of $\mathcal{A}$ in $C_1$, we can prove that the maximum distance between two consecutive nodes in $I_1$ is small and therefore, the maximum length of a gap in $I$ is small. Finally, we can run a greedy sequential $\MIS$ algorithm to fill the gap between every two consecutive nodes in $I$ and find an $\MIS$ in $C$. Hence, if we can find an approximate-$\MaxIS$ in $C_1$ in $o(\log^* |C_1|)$ rounds, then we can find an $\MIS$ in $C$ in $o(\log^*(2^{|C|}))=o(\log^*|C|)$ rounds, contradicting Naor's lower bound (Theorem~\ref{fact:mis-lb}). An illustration of the reduction with all the steps is provided in Figure 1 in Section~\ref{sec:LB}. A detailed reduction is provided in Section~\ref{sec:LB}, together with the full proof of the lower bound. 

\section{Preliminaries}\label{sec:Prelim}

Some of our proofs use the following standard probabilistic tools. An excellent source for the following concentration bounds is the book by Alon and Spencer~\cite{AlonS92}. These bounds can also be found in many lecture notes about basic tail and concentration bounds.

\begin{fact}\label{fact:chrnoff}(\textbf{Multiplicative Chernoff Bound}).
	 Let $X_1, ..., X_n$ be independent random variables taking values in $\{0, 1\}$. Let $X$ denote their sum and let $\mu = E[X]$ denote the sum's expected value. Then for any $0\leq \epsilon\leq 1$, it holds that: $$Pr[|X - \mu| \geq \epsilon \mu] \leq 2 \exp\left(-\frac{\epsilon^2}{2+\epsilon}\mu\right)$$
\end{fact}

\begin{fact}\label{fact:Berns}(\textbf{Bernstein's Inequality}).
	Let $X_1, ..., X_n$ be independent random variables such that $\forall i,  X_i\leq M$. Let $X$ denote their sum and let $\mu = \mathrm{E}[X]$ denote the sum's expected value. Then for any positive $t$, it holds that: $$Pr[|X - \mu| \geq t] \leq 2 \exp\left(-\frac{t^2/2}{Mt/3+\sum_{i=1}^{n} \mathrm{Var}(X_i)}\right)$$
\end{fact}
\begin{fact}\label{fact:Azuma}(\textbf{One-sided Azuma's Inequality}).
	Suppose $\{X_i: i=0,1,2,\hdots\}$ is a martingale and that $|X_i - X_{i-1}|\le c_i$ almost surely. Then, for all positive integers $N$ and all positive reals $t$,
	
	$$\Pr[X_N - X_0 \le -t] \le \exp\left(-\frac{t^2}{2\sum_{i=1}^N c_i^2}\right)$$
\end{fact}

In Section~\ref{sec:arb} we show an improved approximation algorithm for a wide range of arboricity. We define the arboricity of a graph, which is denoted by $\alpha$, as follows~\cite{wiki:xxx}.

\begin{definition}\label{def:arb}
	Given a graph $G$, Let $m_H$ and $n_H$ be the number of edges and nodes of a subgraph $H$ of $G$, respectively. The arboricity of $G$ is $\alpha=\max_{H\subseteq G} \left\lceil \frac{m_H}{n_H-1}\right\rceil$.
\end{definition}

\paragraph{Assumptions:} In all of our upper and lower bounds, we don't assume that the nodes have any global information. In particular, they don't know $n$ or $\Delta$. The only information that each node has before the algorithm starts is its own identifier, and some polynomial upper bound on $n$ (Since the nodes can send $c\log n$ bits in each round to each of their neighbors, naturally, they know some polynomial upper bound on $n$).

\paragraph{Some notations:} The input graph is denoted by $G_w=(V,E,w)$, where $V$ is the set of nodes, $E$ is the set of edges, and $w$ is the weight function. The reason that we choose to add the weight function in a subscript is that some parts of the analysis deal with graphs that have the same sets of nodes and edges as the input graph, but a different weight function. Hence, such a graph will be denoted by $G_{w'}=(V,E,w')$, to indicate that it is the same as the input graph, but with weight function $w'$ rather than $w$. 

We denote by $N^+(v)$ the inclusive neighborhood of $v$, which consists of $N(v)\cup \{v\}$, where $N(v)$ is the set of neighbors of $v$. Furthermore, we denote by $deg(v)=|N(v)|$ the number of neighbors of a node $v$. Finally, we denote by $w(V')$ the total weight of nodes in $V'\subseteq V$. That is, $w(V')=\sum_{v\in V'} w(v)$.

\section{A $\boldsymbol{(1+\epsilon)\Delta}$-Approximation Algorithm}\label{sec:C}

In this section we prove Theorems~\ref{thm:main} and~\ref{thm:Wmain}.

\begin{theorem-repeat}{thm:main}
	\ThmMainDApp
\end{theorem-repeat}

\begin{theorem-repeat}{thm:Wmain}
	\ThmMainAppEx
\end{theorem-repeat}

Theorems~\ref{thm:main} and~\ref{thm:Wmain} share a similar proof structure. First, we present algorithms for $O(\Delta)$-approximation in Sections~\ref{sec:warmup} and~\ref{sec:WeightedSparsifier}. Then, by using a general boosting theorem (Theorem~\ref{thm:amplification} in Section~\ref{sec:amplification}), we get $(1+\epsilon)\Delta$-approximation algorithms.

\subsection{An $\boldsymbol{O(\MIS(n,\Delta))}$-Round Algorithm for $\boldsymbol{O(\Delta)}$-Approximation}\label{sec:warmup}
In this section we show a very simple $O(\MIS(n,\Delta))$-round algorithm that finds an $O(\Delta)$-approximation for $\MaxIS$.
\begin{theorem}\label{thm:mainO1}
	\ThmMainD
\end{theorem}

\paragraph{Algorithm} For every $v\in V$, let $\delta(v)$ be the maximum degree of a node in the inclusive neighborhood of $v$. That is, $\delta(v)=\max \{deg(u)\mid u\in N^+(v)\}$. A node $v$ is called \emph{good} if $w(v)\geq \frac{1}{2(\delta(v)+1)}\sum_{u\in N^+(v)} w(u)$. The algorithm finds a maximal independent set $I$ in the subgraph induced by the set of good nodes. We prove the following lemma.

\begin{lemma}\label{lem:weighted}
	$w(I)\geq w(V)/4(\Delta+1)$
	\begin{proof}
		Let $V_{good}$ be the set of good nodes, and let $\overline{V}=V\setminus V_{good}$. Observe that,
		\begin{align*}
		&\sum_{v\in \overline{V}} w(v)\leq \sum_{v\in \overline{V}} \frac{1}{2(\delta(v)+1)}\sum_{u\in N^+(v)} w(u)\leq\sum_{v\in V} \frac{deg(v)+1}{2(deg(v)+1)} w(v)= w(V)/2\\
		\Rightarrow&\sum_{v\in I} w(v)\geq \sum_{v\in I} \frac{1}{2(\delta(v)+1)}\sum_{u\in N^+(v)}w(u)\geq \sum_{v\in I}\frac{1}{2(\Delta+1)}\sum_{u\in N^+(v)\cap V_{good}}w(u)\\
		&\geq \frac{1}{2(\Delta+1)}\sum_{v\in V_{good}}w(v)\geq w(V)/4(\Delta+1)
		\end{align*}
		as desired. Since the value of an optimal solution in $G_w$ is at most $w(V)$, the algorithm returns an $O(\Delta)$-approximation for $\MaxIS$.
	\end{proof}
\end{lemma}

\paragraph{Success with high probability:} Given a graph of $n$ nodes, an algorithm that finds a maximal independent set in the graph with high probability is an algorithm that succeeds with probability at least $1-1/n^c$ for some constant $c>1$. In the algorithm above, the black box can be a randomized algorithm that is run on a subgraph $H=(V_H,E_H)$ of $G_w$. Since $n_H=|V_H|$ is potentially much smaller than $n$, one may wonder whether the algorithm above actually succeeds with high probability with respect to $n$. The main idea is to use an algorithm that is \emph{intended} to work for graphs with $n$ nodes, rather than $n_H$ nodes. We prove the following lemma, whose proof is by a simple padding argument that is deferred to  Appendix~\ref{app:warm}.

\begin{lemma}\label{lem:MISSIZE}
	Let $\mathcal{A}$ be an $\MIS(n,\Delta)$-round algorithm that finds a maximal independent set with success probability $p$ in a graph of $n$ nodes and maximum degree $\Delta$. Let $H=(V_H,E_H)$ be a graph of $n_H\leq n$ nodes with $(c\log n)$-bit identifiers, for some constant $c$, and let $\Delta_H$ be the maximum degree in $H$. There is an $O(\MIS(n,\Delta_H))$-round algorithm $\mathcal{A'}$ that finds a maximal independent set in $H$ with success probability $p$.
\end{lemma} 

%Since the algorithm used to prove Theorem~\ref{thm:main} is an $\MIS$-based algorithm, using Lemma~\ref{lem:MISSIZE}, we can generalize Theorem~\ref{thm:main} for graphs with number of nodes less than $n$. Specifically, we obtain the following lemma as a corollary of Lemma~\ref{lem:MISSIZE} and Theorem~\ref{thm:main}, which we use as a black-box in the following subsections. 

%\begin{lemma}\label{lem:useMain}
%	Given a weighted graph $H=(V_H,E_H,w_H)$ of $n_H\leq n$ nodes and maximum degree $\Delta_H$. There is an $O(\MIS(n,\Delta_H))$-round algorithm that finds an independent set in $H$ of total weight at least $W(V_H)/4(\Delta_H+1)$, with success probability $1-1/\poly(n)$.

\subsection{A $\boldsymbol{\poly(\log\log n)}$-Round Algorithm for $\boldsymbol{O(\Delta)}$-Approximation}\label{sec:WeightedSparsifier}
In this section we show a $\poly(\log\log n)$-round algorithm that finds an $O(\Delta)$-approximation.

\begin{theorem}\label{thm:WmainO2}
	\ThmMain
\end{theorem}

Our algorithm has the following two-step structure.

\begin{enumerate}
	\item First, we sample a sparse subgraph $H_w=(V_H,E_H,w)$ of $G_w$ with the following two properties:
	\begin{enumerate}
		\item The maximum degree $\Delta_H$ of $H_w$ is at most $O(\log n)$.
		\item $w(V_H)/\Delta_H=\Omega( w(V)/\Delta)$. That is, the ratio between the total weigh and maximum degree in $H_w$ is at least, up to a constant factor, as in $G_w$.
	\end{enumerate} 
	\item Then, we use Theorem~\ref{thm:mainO1} to find an independent set in $H_w$ of size at least  $\frac{w(V_H)}{4(\Delta_H+1)}=\frac{w(V)}{c\Delta}$, for some constant $c>1$, in $O(\MIS(n,\Delta_H))=O(\MIS(n,\log n))=\poly(\log\log n)$ rounds.
\end{enumerate} 

%\paragraph{Main idea of the sampling:} In the unweighted case, sampling such a subgraph is almost trivial: Each node joins the subgraph with probability $\min\{\log n/\Delta,1\}$. One can show by standard Chernoff and Union-Bound arguments that the maximum degree of the subgraph is at most $O(\log n)$, and that the number of nodes is at least $\min\{n/2,n\log n/\Delta\}$, with high probability. 

\paragraph{The first step: sampling a subgraph with the desired properties.} Recall that $w(N(v))$ is the sum of weights of the neighbors of $v$, which we call the \emph{weighted degree} of $v$. For each node $v\in V$, let  $w_{max}(v)=\max\{w(N(u))\mid u\in N^+(v)\}$. It is useful to think about $w_{max}(v)$ as the \emph{maximum weighted degree} of a node in the inclusive neighborhood of $v$. We sample a subgraph $H_w=(V_H,E_H,w)$, as follows. Let $\lambda\geq 1$ be a constant to be chosen later. Recall that $\delta(v)$ is the maximum degree of a node in the inclusive neighborhood of $v$. Each node $v\in V$ joins $V_H$ with probability $$p(v)=\min\{\lambda\log n\cdot (\frac{1}{\delta(v)}+\frac{w(v)}{w_{max}(v)}),1\} $$

In Lemma~\ref{lem:DegreeW}, we show that the maximum degree of $H_w$ is $\Delta_H=O(\log n)$. In Lemma~\ref{lem:WeightedSize}, we show that $w(V_H)=\Omega(\min\{w(V),w(V)\log n/\Delta\})$.
%\Scomment{This should not be a footnote}\footnote{Perhaps the first thing that comes into mind is to try to sample each node with probability $\frac{\log n\cdot w(u)}{W_{max}(N^+(u))}$, as this is the natural extension of $\frac{\log n}{d_{max}(N^+(u))}$, which works for the unweighted case. However, it turns out that sampling with probability $\frac{\log n\cdot w(u)}{W_{max}(N^+(u))}$ might result in a subgraph of a total weight $o(\log n\cdot W(V)/\Delta)$, which is too small for our purposes. It turns out that in order to get around this issue, it suffices to boost this sampling probability by an additive factor of $\frac{\log n}{d_{max}(N^+(u))}$.} 

\begin{lemma}\label{lem:DegreeW}The maximum degree $\Delta_H$ in $H_w$ is $O(\log n)$, with high probability.
\end{lemma}
\begin{proof}
	Let $V^+=\{v\in V\mid p(v)\geq 1\}$. We show that each node $u$ has at most $O(\log n)$ neighbors in $V^+\cap V_H$, and at most $O(\log n)$ neighbors in $(V\setminus V^+)\cap V_H$. Let $N_H(v)$ be the set of neighbors of $v$ in $H$.
	\begin{enumerate}
		\item For every $v\in V$, $|N_H(v)\cap V^+|\leq 2\lambda\log n$:  Assume towards a contradiction that there are more than $2\lambda\log n$ nodes in $N_H(v)\cap V^+$. Since each node $v\in V^+$ has $p(v)\geq 1$. it holds that \begin{align*}
		\sum_{u\in N(v)\cap V^+} p(u)\geq \sum_{u\in N_H(v)\cap V^+}p(u)> 2\lambda\log n
		\end{align*}
		On the other hand, \begin{align*}\sum_{u\in N(v)\cap V^+} p(u)\leq \sum_{u\in N(v)}p(u)=\sum_{u\in N(v)}\lambda\log n\cdot (\frac{1}{\delta(v)}+\frac{w(v)}{w_{max}(v)})
		\end{align*}
		Since $deg(v)=|N(v)|$ and $w(N(v))=\sum_{u\in N(v)} w(u)$ are lower bounds on $\delta(v)$ and $w_{max}(v)$, respectively, we have that \begin{align*}
		\sum_{u\in N(v)}\lambda\log n\cdot (\frac{1}{\delta(u)}+\frac{w(u)}{w_{max}(u)})\leq \sum_{u\in N(v)}\lambda\log n\cdot(\frac{1}{deg(v)}+\frac{w(u)}{w(N(v))})=2\lambda\log n
		\end{align*}
		which is a contradiction.
		\item $|N_H(v)\cap (V\setminus V^+)|\leq 2\lambda\log n$: Observe that the expected number of neighbors of $v$ in $N_H(v)\cap (V\setminus V^+)$ is \begin{align*}
		\sum_{u\in N(v)} p(u)\leq 2\lambda\log n
		\end{align*}
		 Since $|N_H(v)\cap (V\setminus V^+)|$ is a sum of independent random variables, one can apply Chernoff's bound (Fact~\ref{fact:chrnoff}) to achieve that this number concentrates around its expectation with high probability.
	\end{enumerate}
	By applying a standard Union-Bound argument over all the nodes, we conclude that the maximum degree in $H_w$ is $\Delta_H=O(\log n)$ with high probability.
\end{proof}

The rest of this section is devoted to the task of proving that $w(V_H)=\Omega(\min\{w(V),w(V)\log n/\Delta\})$. This is proved in Lemma~\ref{lem:WeightedSize}. First, we start by proving a slightly weaker lemma, that assumes that for all $v\in V$, $p(v)\leq 1$.

\begin{lemma}\label{lem:WeightedSizeWeaker}
	Assume $p(v)\leq 1$, for all $v\in V$. It holds that $w(V_H)=\Omega(w(V)\log n/\Delta)$, with high probability.
\end{lemma}

\paragraph{Main idea of the proof of Lemma~\ref{lem:WeightedSizeWeaker}:} Let $w_1\geq w_2\geq...\geq w_n$ be a sorting of the weights of nodes in $V$ in a decreasing order (where ties are broken arbitrarily). Let $V_{high}=\{u\in V\mid w(u)\in \{w_1,...,w_{\Delta}\}\}$, and let $V_{low}=V\setminus V_{high}=\{u\in V\mid w(u)\in \{w_{\Delta+1},...,w_n\}\}$. That is, $V_{high}$ contains the $\Delta$ heaviest nodes, and $V_{low}$ contains all the other nodes.  The proof is split into the following two cases that are proven separately in Claims~\ref{claim:FirstCase} and~\ref{claim:SecondCase}.
\begin{enumerate}
	\item $w(V_{high})\geq w(V)/2$: In this case, at least half of the total weight is distributed among high-weight nodes. Intuitively, we need to make sure that we get many of these high-weight nodes. Since the number of high-weight nodes that are sampled is a sum of independent random variables, we are able to use Chernoff's bound to prove that many of them are sampled, with high probability. The full proof for this case is presented in Claim~\ref{claim:FirstCase}.
	\item $w(V_{low})\geq w(V)/2$: In this case, at least half of the total weight is distributed among low-weight nodes. Therefore, it is sufficient to show that $w(V_H)=\Omega(w(V_{low})\log n/\Delta)$. The key property here is that we can bound the maximum weight of a node in $V_{low}$ by $w(V)/\Delta$. We show how to use this property together with Bernstein's inequality to prove Lemma~\ref{lem:WeightedSizeWeaker} for this case. The full proof for this case is presented in Claim~\ref{claim:SecondCase}.
\end{enumerate}

\begin{claim}\label{claim:FirstCase}
	Assume that for all $v\in V$,  $p(v)\leq 1$. Let $V_{high}=\{u\in V\mid w(u)\in \{w_1,...,w_{\Delta}\}\}$. If $w(V_{high})\geq w(V)/2$, then $w(V_H)=\Omega(w(V)\log n/\Delta)$, with high probability.
	\begin{proof}
		Let $S=\{v\in V_{high}\mid w(v)\geq w(V)/4\Delta\}$. We start by showing that at least a constant fraction of the total weight in $G_w$ is distributed among nodes in $S$. Let $\overline{S}=V_{high}\setminus S$, we start by showing that $w(\overline{S})\leq w(V)/4$: 
		
		\begin{align*}
		w(\overline{S})\leq \sum_{v\in \overline{S}} w(v)\leq \sum_{v\in \overline{S}} \frac{w(V)}{4\Delta}\leq \frac{w(V)}{4} 
		\end{align*}

		where the last inequality holds because $|\overline{S}|\leq |V_{high}|=\Delta$. Therefore, $w(S)=w(V_{high}\setminus \overline{S})= w(V_{high})-w(V(\overline{S}))\geq w(V)/4$. Next, we show that $|S\cap V_H|=\Omega(\log n)$, by using  Chernoff's bound. Let $x_v$ be a $\{0,1\}$ random variable indicating whether $v\in V_H$, and let $X=\sum_{v\in S}x_v$. We show that the expectation of $X$ is at least $c\log n/4$.  \begin{align*}\mathbb{E}[X]&=\sum_{v\in S}\mathbb{E}[x_v]= \sum_{v\in S}p(v)=\sum_{v\in S}\lambda\log n\cdot (\frac{1}{\delta(v)}+\frac{w(v)}{w_{max}(v)})\\
		&\geq \sum_{v\in S}\frac{w(v)\lambda\log n}{w(V)}\geq \frac{\lambda\log n}{w(V)}\cdot \sum_{v\in S}w(v)=\frac{w(S)\lambda\log n}{w(V)}\geq \frac{\lambda\log n}{4}
		\end{align*}
		
		Furthermore, sine $X$ is a sum of independent $\{0,1\}$ random variables with expectation $\Omega(\log n)$, by applying Chernoff's bound (Fact~\ref{fact:chrnoff}), we conclude that there are at least $\Omega(\log n)$ nodes in $S\cap V_H$, with high probability. Since each node in $S$ has weight at least $w(V)/4\Delta$, this implies that the total weight in $V_H$ is $w(V_H)\geq w(S\cap V_H)= \Omega(w(V)\log n/\Delta)$, with high probability, as desired.
	\end{proof}
\end{claim}

\begin{claim}\label{claim:SecondCase}
	Assume that for all $v\in V$,  $p(v)\leq 1$. Let $V_{low}=\{v\in V\mid w(v)\in \{w_{\Delta+1},...,w_n\}\}$. If $w(V_{low})\geq w(V)/2$, then $w(V_H\cap V_{low})=\Omega(w(V)\log n/\Delta)$, with high probability.
	\begin{proof}
	 Let $x_v$ be a $\{0,1\}$ random variable indicating whether $v\in V_H$, let $y_v=x_v\cdot w(v)$, and let $Y=\sum_{v\in V_{low}}y_v$. We prove the following 3 properties:
	 \begin{enumerate}
		\item $\mathbb{E}(Y)\geq \frac{w(V)\lambda\log n}{2\Delta}$: this is because 
	 
		 \begin{align*}\mathbb{E}[Y]&=\sum_{v\in V_{low}}p(v)\cdot w(v)=\sum_{v\in V_{low}}\lambda\log n\cdot (\frac{1}{\delta(v)}+\frac{w(v)}{w_{max}(v)})\cdot w(v)\\
		 &\geq \sum_{v\in V_{low}} \frac{w(v)\lambda\log n}{\Delta}= \frac{w(V_{low})\lambda\log n}{\Delta}\geq \frac{w(V)\lambda\log n}{2\Delta}
		 \end{align*}
	 
		 where the last equality holds since $w(V_{low})\geq w(V)/2$.

	 	\item For any $v\in V_{low}$, it holds that $w(v)\leq w(V)/\Delta$: Recall that $\{w_1,\cdots,w_n\}$ is an ordering of the weight of nodes by a decreasing order. Hence, for any $j$, it holds that
	 	\begin{align*}&w_j\cdot j\leq \sum_{i=1}^{j} w_j\leq w(V)\end{align*}
	 	
	 	where the first inequality holds because $w_j$ is the minimum among $\{w_1,...,w_j\}$. Hence, since each node $v\in V_{low}$ has weigh $w_j$ where $j>\Delta$, we have that $w(v)\leq w(V)/\Delta$ for any $v\in V_{low}$.

	 	\item It holds that $\sum_{v\in V_{low}}\mathbb{E}[y_v^2]\leq w(V)\cdot \mathbb{E}[Y]/\Delta$: First, observe that \begin{align*}
	 	&\sum_{v\in V_{low}}\mathbb{E}[y_v^2]\leq \max\{w(v)\mid v\in V_{low}\}\cdot \sum_{v\in V_{low}}\mathbb{E}[y_v]=\max\{w(v)\mid v\in V_{low}\}\cdot\mathbb{E}[Y]\\
	 	&\leq\frac{w(V)\cdot \mathbb{E}[Y]}{\Delta}
	 	\end{align*}
	 	where the last inequality holds by the second property.
	 	\end{enumerate}
 	
	 	By proving these three properties, we have satisfied all the prerequisites of Bernstein's inequality. A direct application of the inequality yields:
	 	
	 	\begin{align*}
	 	Pr\big[|Y - \mathbb{E}[Y]| \geq \mathbb{E}[Y]/2\big] \leq 2 \exp\left(-\frac{\mathbb{E}[Y]^2/8}{M\cdot \mathbb{E}[Y]/6+\sum_{v\in V_{low}} \mathrm{Var}(y_v)}\right)
	 	\end{align*}
		By the second and third properties, we have that
		\begin{align*}
		&\sum_{v\in V_{low}} \mathrm{Var}(y_v)= \sum_{v\in V_{low}} \mathbb{E}(y_v^2)-\mathbb{E}[y_v]^2\leq\sum_{v\in V_{low}} \mathbb{E}(y_v^2)\leq  \frac{w(V)\cdot \mathrm{E}[Y]}{\Delta}\\
		\Rightarrow& Pr\big[|Y - \mathbb{E}[Y]| \geq \mathbb{E}[Y]/2\big] \leq 2 \exp\left(-\frac{\mathbb{E}[Y]^2/8}{\frac{w(V)\cdot \mathbb{E}[Y]}{6\Delta}+\frac{w(V)\cdot \mathbb{E}[Y]}{\Delta}}\right)\leq 2 \exp\left(-\frac{6\Delta\cdot \mathbb{E}[Y]/8}{7w(V)}\right)
		\end{align*}

	 	Furthermore, by the first property, we have that \begin{align*}
	 	&\mathbb{E}[Y]\geq w(V)\lambda\log n/2\Delta\\
	 	\Rightarrow& 2 \exp\left(-\frac{6\Delta\cdot \mathbb{E}[Y]/8}{7w(V)}\right)\leq 2 \exp\left(-\frac{6 w(V)\lambda\log n}{56w(V)}\right)=2 \exp\left(-\frac{6\lambda\log n}{56})\right)
	 	\end{align*}

	 	 Finally, choosing $\lambda=112/6$ implies that:
	 	
	 	\begin{align*}
	 	Pr\big[|Y - \mathbb{E}[Y]| \geq \mathbb{E}[Y]/2\big]\leq \frac{1}{n^{2\log e}}<\frac{1}{n^{2}}
	 	\end{align*}
	 	
	 	as desired. Furthermore, we can boost the success probability to $1-1/n^k$ for any constant $k>1$, by setting $\lambda=\frac{112k}{3}$.
		
	\end{proof}
\end{claim}

Having proved claims~\ref{claim:FirstCase} and~\ref{claim:SecondCase}, this finishes the proof of Lemma~\ref{lem:WeightedSizeWeaker}. Lemma~\ref{lem:WeightedSizeWeaker} makes the assumption that $p(v)\leq 1$ for all $v\in V$. We remove this assumption in the proof of the following lemma. 

\begin{lemma}\label{lem:WeightedSize}
 It holds that $w(V_H)=\Omega(\min\{w(V),w(V)\log n/\Delta\})$, with high probability.
 \begin{proof}
 	Let $V^+=\{u\in V\mid p(w)\geq 1\}$. The proof is split into two cases:
 	\begin{enumerate}
 		\item $w(V^+)\geq w(V)/2$: Since all the nodes in $V^+$ join $V_H$ deterministically, this implies that $w(V_H)\geq w(V^+)\geq w(V)/2$.
 		\item $w(V^+)<w(V)/2$: This implies that $w(V\setminus V^+)\geq w(V)/2$. Since each node $w\in V\setminus V^+$ has $p(w)<1$, we can apply Lemma~\ref{lem:WeightedSizeWeaker} directly on the nodes in $V\setminus V^+$ to conclude that $w(V_H)=\Omega(w(V\setminus V^+)\log n/\Delta)=\Omega(w(V)\log n/\Delta)$, with high probability, as desired.
 	\end{enumerate}
 \end{proof}
\end{lemma}

Now we are ready to finish the proof of Theorem~\ref{thm:WmainO2}.

\begin{proof}[\textbf{Proof of Theorem \ref{thm:WmainO2}}]
	Since both Lemma~\ref{lem:DegreeW} and~\ref{lem:WeightedSize} above hold with high probability, we can apply another standard Union-Bound argument to conclude that both of them hold with high probability (simultaneously). Hence, by running the algorithm from Section~\ref{sec:warmup} on $H_w$, we get an independent set of weight $\Omega(w(V_H)/\Delta_H)=\Omega(\min\{w(V),w(V)\log n/\Delta\}/\Delta_H)=\Omega(w(V)/\Delta)$, in $\MIS(n,\Delta_H)=\MIS(n,\log n)=\poly(\log\log n)$ rounds, with high probability, as desired.
\end{proof}
%Plugging in the best deterministic known maximal independent set algorithm for the CONGEST model results in an $O(2^{O(\sqrt{\log\log n \log\log\log n})})$
%\fi
\subsection{Boosting for $\boldsymbol{(1+\epsilon)\Delta}$-Approximation}\label{sec:amplification}

In this section we prove the following theorem. 
\begin{theorem}\label{thm:amplification}
	Let $\mathcal{A}$ be a $T$-round algorithm that finds an independent set of weight at least $(\frac{1}{c\Delta})$-fraction of the total weight in the graph, in the CONGEST model. There is a $(\frac{2Tc}{\epsilon})$-round algorithm $\mathcal{A}'$ that finds a $(1+\epsilon)\Delta$-approximation for maximum-weight independent set in the CONGEST model.
\end{theorem}

\paragraph{Description of algorithm $\mathcal{A}'$:}
Our algorithm consists of two stages. In the first stage we iteratively call algorithm $\mathcal{A}$, which returns an independent set, and perform \emph{weight reductions} based on the independent set. Then, we push all nodes in the independent set onto a stack. The next time we invoke $\mathcal{A}$ in the first stage, it is invoked on the graph with the \emph{reduced} weight. In the second stage we iteratively pop the independent sets from the stack and greedily construct another independent set which we return as the solution.
$\mathcal{A}'$ is formally given in Algorithm~\ref{alg:amplify}. We continue with a formal description of our algorithm.

Let $t=c/\epsilon$. As stated before, there are two stages in $\mathcal{A}'$, each consists of $t$ phases. 
\paragraph{First Stage:} For $i=1$ to $t$ phases, in each phase, we first run algorithm $\mathcal{A}$ to find an independent set $I_i$ in $G_{w_i}$, where $w_1=w$, and $G_{w_1}=G_w=(V,E,w)$ is the original input graph. Then, we insert the nodes in $I_i$ to a stack $S$ that is initially defined to be empty, and continue to $G_{w_{i+1}}=(V,E,w_{i+1})$, where $w_{i+1}$ is defined as follows. For each $v\in V$,
\begin{align*}
&w_{i+1}(v)=\begin{cases}
0 & \mbox{if } v\in I_i\\
w_i(v)-\sum_{u\in N(v)\cap I_i}w_i(u) & \mbox{otherwise}
\end{cases}
\end{align*}

That is, $w_{i+1}$ is a weight function which results from weight reductions to $w_i$, as follows. For each $v\in I_i$, we reduce its total current weight, $w_i(v)$, and therefore its weight becomes zero. For each $v\notin I_i$, we reduce its current weight, $w_{i}(v)$, by the total weight of its neighboring nodes in $I_i$. This concludes the first stage. Before we proceed to the second stage, let us define the following weight functions $w'_i$, for every $i\in [t]$, that are used in the analysis. For each $v\in V$,
\begin{align*}
w'_i(v)=w_i(v)-w_{i+1}(v)
\end{align*}
Hence, $w'_i(v)$ is the \emph{reduced} weight from $v$ at the end of phase $i$. We define another weight function $w'$, which is the \emph{total reduced} weight function. For each $v\in V$,
\begin{align*}
w'(v)=\sum_{i=1}^{t} w'_i(v)
\end{align*}

\paragraph{Second Stage:} We construct an independent set $I$ as follows. For $i=0$ to $t-1$ phases, we pop out $I_{t-i}$ from the stack, and insert each $v\in I_{t-i}$ to $I$ unless $I$ already contains a neighbor of $v$. 

In Lemma~\ref{lem:amplify}, we prove that $I$ is a $(1+\epsilon)\Delta$-approximation for maximum-weight independent set for $G_w$. First, let us start with the following helper propositions. Recall that given a weight function $\hat{w}$ we denote by $G_{\hat{w}}=(V,E,\hat{w})$ the same graph as the input graph $G_w$ but with weight function $\hat{w}$ rather than $w$. For every set of nodes, $V' \subseteq V$, we define $\hat{w}(V')=\sum_{v\in V'} \hat{w}(v)$.

\begin{proposition}\label{prop:IndPhase}
	$I_i$ is a $\Delta$-approximation for maximum-weight independent set in $G_{w'_i}$.
\end{proposition}
\begin{proof}
	First, observe that for every $v\in I_i$, it holds that $w'_i(v)=w_i(v)-w_{i+1}(v)=w_i(v)$, and for every $v\notin I_i$, it holds that $w'_{i}(v)=w'_i(N(v)\cap I_i)$.
	%\begin{align*}
	%w'_{i}(v)= w_i(v)-w_{i+1}(v)\leq w_i(v)-\left(w_i(v)- %w_i(N(v)\cap I_i)\right)= w'_i(N(v)\cap I_i)
	%\end{align*}
	%Therefore, the total amount of weight in $G_{w'_i}$ is at most $(\Delta+1)\sum_{v\in I_i} w'_i(v)$. 
	Let $I^*_i$ be an optimal maximum-weight independent set in $G_{w'_i}$. We can assume that $I^*_i$ contains only nodes in $I_i\cup N(I_i)$, as all the other nodes in $G_{w'_i}$ have zero weight\footnote{Recall that $N(I_i)$ denotes the set of neighbors of all nodes in $I_i$.}. We have that,
	\begin{align*}
	&w'_i(I^*_i)= w'_i(I^*_i\cap I_i)+ w'_i(I^*_i\setminus I_i) = w'_i(I^*_i\cap I_i)+\sum_{v\in I^*_i\setminus I_i} w'_i(N(v)\cap I_i)\\
	&= w'_i(I^*_i\cap I_i) + \Delta w'_i(I_i\setminus I^*_i) \leq \Delta w'_i(I_i)
	\end{align*}
	as desired.
\end{proof}
In the following proposition we draw a connection between $w(I)$, the final value of our solution, and the total value stored in our stack. Formally, we show the following.
%	show that the total weight of $I$, the independent set genera}
\begin{proposition}\label{prop:stack}
	It holds that $w(I)\geq\sum_{i=1}^{t} w'_i(I_i)=\sum_{i=1}^{t} w_i(I_i)$.
	\begin{proof}
		We assume without loss of generality that $\mathcal{A}$ never picks nodes of non-positive weight to the independent set, as we can always remove them and increase the size of the solution. 
		%This assumption is indeed without loss of generality because nodes of non-positive weight can only hurt the quality of the solution. 
		Observe that for every $v\in I$, it holds that $v\in I_i$ for some $i\in [t]$. Let $i_v$ be the phase for which $v\in I_{i_v}$. It holds that,
		\begin{align*}
		w(v)=w_{i_v}(v)+\sum_{i=1}^{i_v-1} w_i(N(v)\cap I_i) = w_{i_v}(v) + \sum_{ u\in(N(v)\cap(\bigcup_{i=1}^{i_v-1}I_i))} w_{i_u}(u)
		\end{align*}
		The first equality is because the weight of $v$ at phase $i_v$ was positive, as otherwise it wouldn't be in $I_{i_v}$, and, until phase $i_v$, the total amount of weight that was reduced from $v$ is $\sum_{i=1}^{i_v-1} w_i(N(v)\cap I_i)$. And the second equality is because for every $u$ which contributes to the sum $\sum_{i=1}^{i_v-1} w_i(N(v)\cap I_i)$, its contribution is exactly $w_{i_u}(u)$.
		Hence, we have that,
		\begin{align*}
		w(I)&=\sum_{v\in I}\left(w_{i_v}(v)+\sum_{i=1}^{i_v-1}  w_i(N(v)\cap I_i)\right) =\left(\sum_{v\in I} w_{i_v}(v)\right)+\left(\sum_{v\in I}\sum_{ u\in(N(v)\cap(\bigcup_{i=1}^{i_v-1} I_i))} w_{i_u}(u)\right)\\
		&\geq\left(\sum_{v\in I} w_{i_v}(v)\right)+\left(\sum_{ u\in(\bigcup_{i=1}^{t} I_i)\setminus I} w_{i_u}(u)\right)=\sum_{i=1}^{t} w_i(I_i)
		\end{align*}
		where the last inequality holds because for every $u\in (\cup_{i=1}^t I_i)\setminus I$, there is at least one neighbor $v$ of $u$ in $I$, with $i_v>i_u$. %The above inequality shows that for every such $u$, we can blame the weight $w_{i_u}(u)$ on $w(u)$.
		Finally, since for every $i\in [t]$ and for every $v\in I_i$, $w'_i(v)=w_i(v)$, the claim follows.
	\end{proof}
\end{proposition}
Now we are ready to show that $I$ is a $(1+\epsilon)\Delta$-approximation for maximum-weight independent set in the original input graph $G_{w}$.
\begin{lemma}\label{lem:amplify}
	$I$ is a $(1+\epsilon)\Delta$-approximation for maximum-weight independent set in $G_{w}$.
	\begin{proof}
		Let $OPT(G_w)$ be the value of an optimal solution in $G_w$. The proof is by the following case analysis.
		\begin{enumerate}
			\item $w_t(V)\leq \frac{\epsilon}{1+\epsilon}OPT(G_w)$:\\
			Recall that for all $v\in V$, $w'(v)=\sum_{i=1}^{t}w'_i(v)$. First, observe that,
			\begin{align*}
			w'(V)&\geq \sum_{v\in V}\sum_{i=1}^{t-1}w'_i(v)=\sum_{v\in V}\sum_{i=1}^{t-1}w_i(v)-w_{i+1}(v)\\
			&=\sum_{v\in V} w_1(v)-w_2(v)+w_2(v)-w_3(v)+\cdots -w_{t-1}(v)+w_{t-1}(v)-w_t(v)\\
			&= w(V)- w_t(V) \geq w(V)-\frac{\epsilon}{1+\epsilon}OPT(G_w)
			\end{align*}
			Therefore, the value of an optimal solution in $G_{w'}$ cannot be very small compared to the value of an optimal solution in $G_{w}$. Namely, $OPT(G_{w'})\geq (1-\frac{\epsilon}{1+\epsilon})OPT(G_{w}) = OPT(G_{w}) / (1+\epsilon)$.
			%\begin{align*}
			%OPT(G_{w'})\geq %(1-\frac{\epsilon}{1+\epsilon})OPT(G_{w})
			%\end{align*}
			Finally, by Propositions~\ref{prop:IndPhase}, $I_i$ is a $\Delta$-approximation to $OPT(G_{w'_i})$, and by Proposition~\ref{prop:stack}, we have that,
			
			\begin{align*}
			&w(I)\geq \sum_{i=1}^t w'_i(I_i)\geq \sum_{i=1}^t \frac{OPT(G_{w'_i})}{\Delta}\geq \frac{OPT(G_{w'})}{\Delta} \geq \frac{OPT(G_{w}) }{(1+\epsilon)\Delta}
			\end{align*} 
			
			as desired.
			\item $ w_t(V)\geq \frac{\epsilon}{1+\epsilon}OPT(G_w)$:\\
			Observe that for any $i<t$, it holds that $w_i(V)\geq w_t(V)$.
			Therefore, since $\mathcal{A}$ returns an independent set of weight at least $(\frac{1}{c\Delta})$-fraction of the total weight in the graph, for each phase $i\in [t]$, it holds that $w_i(I_i)\geq \frac{\epsilon}{(1+\epsilon)c\Delta}OPT(G_w)$, which implies:
			\begin{align*}
			\sum_{i=1}^{t} w_i(I_i)\geq t\frac{\epsilon}{(1+\epsilon)c\Delta}OPT(G_w)=\frac{1}{(1+\epsilon)\Delta}OPT(G_w)
			\end{align*}
			Finally, by Proposition~\ref{prop:stack}, we have that $w(I)\geq \sum_{i=1}^t  w_i(I_i)\geq \frac{1}{(1+\epsilon)\Delta}OPT(G_w)$. Which completes the proof.
			
			%\begin{align*}
			%w(I)\geq \sum_{i=1}^t %w'_i(I_i)=\sum_{i=1}^t  w_i(I_i)\geq %\frac{1}{(1+\epsilon)\Delta}OPT(G_w)
			%\end{align*}
			%as desired.
		\end{enumerate}    
	\end{proof}
\end{lemma}

\paragraph{Remark:}
One can show that the same algorithm also returns an independent set of weight at least $\frac{w(V)}{(1+\epsilon)(\Delta+1)}$, which is sometimes better than $(1+\epsilon)\Delta$-approximation, and sometimes worse. The proof of this argument similar to the proof above. The only difference is that the case analysis in Lemma~\ref{lem:amplify} is with respect to $\frac{\epsilon}{1+\epsilon}w(V)$. That is, the first case is in which $w_t(V)\leq \frac{\epsilon}{1+\epsilon}w(V)$, and the second case is in which $w_t(V)\geq \frac{\epsilon}{1+\epsilon}w(V)$. It is straightforward to show that in both cases $I$ is of weight at least $\frac{w(V)}{(1+\epsilon)(\Delta+1)}$, by using the stack property (Proposition~\ref{prop:stack}).

\begin{corollary}\label{cor:boost}
		Let $\mathcal{A}$ be a $T$-round algorithm that finds an independent set of weight at least $(\frac{1}{c\Delta})$-fraction of the total weight in the graph, in the CONGEST model. There is a $(\frac{2Tc}{\epsilon})$-round algorithm $\mathcal{A}'$ that finds an independent set $I$ of weight at least $\frac{w(V)}{(1+\epsilon)(\Delta+1)}$ in the CONGEST model.
\end{corollary}

\begin{algorithm}[]\label{alg:amplify}
	\SetAlgoLined
	\DontPrintSemicolon
	\KwData{a graph $G=(V,E,w)$}
	\KwResult{an independent set $I$}
	
	$I\gets \emptyset$\\
	$S\gets \emptyset$\tcp*{$S$ is the Stack}
	$w_1=w$\\
	
	\For{$i=1$ to $t=c/\epsilon$ phases}{
		
		Run $\mathcal{A}$ on $G_{w_i}$\\
		Let $I_i\gets \mathcal{A}(G_{w_i})$\\
		Insert all the nodes in $I_i$ into $S$\\
		$\forall v\in V: w_{i+1}(v)\gets w_{i}(v)-\sum_{u\in N^+(v)\cap I_i} w_i(u)$\\
		
	}
	\For{$i=0$ to $t-1$ phases}{
		
		Pop $I_{t-i}$ from $S$\\
		%$I=I\cup (I_{t-i} \setminus N(I))$\\
		
		\For{$v\in I_{t-i}$}{  \If{$N(v)\cap I=\emptyset$}
			{add $v$ to $I$}
		}
	}
	\Return $I$\\
	\caption{$(1+\epsilon)\Delta$-approx}
\end{algorithm}

\section{Faster Algorithm for Low-Degree Graphs}\label{sec:low-deg}

In this section we show that there is an $O(1)$-round algorithm that finds an independent set of size  $\Omega(n/\Delta)$ for graphs in which $\Delta\le n/\log n$. Using the boosting theorem that is presented in the previous section (in particular, Corollary~\ref{cor:boost}), this implies that there is an $O(1/\epsilon)$-round algorithm that finds an independent set of size at least $\frac{n}{(1+\epsilon)(\Delta+1)}$ in unweighted graphs of maximum degree $\Delta\leq n/\log n$, proving Theorem~\ref{thm:LD}.

The algorithm is based on a new analysis of one round of the classical ranking algorithm for independent set, which is (to the best of our knowledge) due to Boppanna. This classical algorithm finds an independent set in a graph by independently selecting a rank for each vertex and including a vertex in the output independent set if its rank is greater than the ranks of its neighbors (Algorithm~\ref{alg:bop}).

\begin{algorithm}[H]\label{alg:bop}
	\SetAlgoLined
	\DontPrintSemicolon
	\KwData{an unweighted graph $G=(V,E)$}
	\KwResult{an independent set $I$}
	
	$I\gets \emptyset$\;
	
	\For{each vertex $u\in V$}{
		
		$r_u\gets $ uniformly random number in $\{1,2,\hdots,100n^{c+2}\}$\;
		
	}
	
	\For{each vertex $u\in V$}{
		
		Add $u$ to $I$ if $r_u > r_v$ for all neighbors $v$ of $u$ in $G$\;
		
	}
	
	\Return $I$\;
	
	\caption{$\ORL(G)$}
\end{algorithm}

In discussion, notice that $\ORL$ can be implemented in $O(c)$ rounds in the CONGEST model. We analyze this algorithm by considering a sequential view. The independent set $I$ returned by the algorithm only depends on the order of the $r_v$s. We could run $\ORL$ instead by picking a uniformly random permutation of the vertices and include a vertex $v$ in $I$ if no neighbor of $v$ has a higher rank in the permutation. Furthermore, we can sample the permutation by repeatedly selecting uniformly random vertices without replacement. Equivalently, sample a permutation by repeatedly selecting vertices with replacement, but reject samples seen before (Algorithm~\ref{alg:bop2}).

\begin{algorithm}[]\label{alg:bop2}
	\SetAlgoLined
	\DontPrintSemicolon
	\KwData{an unweighted graph $G=(V,E)$}
	\KwResult{an independent set $I$}
	
	$I\gets \emptyset$\;
	
	$U\gets V$\;
	
	\While{$U\ne\emptyset$}{
		
		$u\gets $ uniformly random element of the set $V$\;
		
		$U\gets U\setminus \{u\}$\;
		
		\If{all neighbors $v$ of $u$ are in $U$}{
			
			Add $u$ to $I$\;
			
		}
		
	}
	
	\Return $S$\;
	
	\caption{$\SORL(G)$}
\end{algorithm}

To analyze $\ORL(G)$, it suffices to analyze $\SORL(G)$:

\begin{proposition}\label{prop:indep-equiv}
	For any unweighted graph $G$ and constant $c > 0$, $\SORL(G)$ produces a distribution over sets $I$ with total variation distance at most $1/n^c$ from the distribution produced by $\ORL(G)$; in particular
	
	$$\sum_{\text{sets } I_0} |\Pr_{I\sim \SORL(G)}[I = I_0] - \Pr_{I\sim \ORL(G)}[I = I_0]|\le 1/n^c$$
\end{proposition}

\begin{proof}
	Let $n = |V|$ and $\mathcal D_0$ be the uniform distribution over $\{1,2,\hdots,100n^{c+2}\}^n$. This is the distribution over rank tuples $(r_u)_{u\in V}$ used by algorithm $\ORL$. Let $\mathcal D_1$ be the uniform distribution over tuples in $\{1,2,\hdots,n\}^n$ with distinct coordinates. Let $\ORL_1$ denote the algorithm with the tuple $(r_u)_{u\in V}$ sampled from $\mathcal D_1$ instead of $\mathcal D_0$:
	
	\begin{algorithm}[H]
		\SetAlgoLined
		\DontPrintSemicolon
		\KwData{an unweighted  graph $G$}
		\KwResult{an independent set $I$}
		
		$I\gets \emptyset$\;
		
		$(r_v)_{v\in V} \gets $ sample from $\mathcal D_1$ \;
		
		\For{each vertex $v\in V$}{
			
			Add $v$ to $I$ if $r_v > r_u$ for all neighbors $u$ of $v$ in $G$\;
			
		}
		
		\Return $I$\;
		
		\caption{$\ORL_1(G)$}
	\end{algorithm}
	
	By a union bound over all pairs of vertices, $r_u\ne r_v$ for all $u,v\in V(G)$ with probability at least $1 - \binom{n}{2}\frac{1}{100n^{c+2}} > 1 - 1/(2n^c)$. Let $E$ denote the event $\{r_u\ne r_v \forall u,v\in V(G)\}$ and let $\overline{E}$ denote the negation. Conditioned on $r_u\ne r_v$ for all $u,v\in V(G)$, the output of $\ORL$ is identically distributed to the output of $\ORL_1$. Therefore,
	
	\begin{align*}
	&\sum_{\text{sets } I_0} |\Pr_{I\sim \ORL(G)}[I = I_0] - \Pr_{I\sim \ORL_1(G)}[I = I_0]|\\
	&=\sum_{\text{sets } I_0} |\Pr_{I\sim \ORL(G)}[I = I_0] - \Pr_{I\sim \ORL(G)}[I = I_0 | E]|\\
	&=\sum_{\text{sets } I_0} |\Pr[I = I_0 | E]\Pr[E] + \Pr[I = I_0 \& \overline{E}] - \Pr[I = I_0 | E]|\\
	&\le \sum_{\text{sets } I_0} (\Pr[I = I_0 | E]\Pr[\overline{E}] + \Pr[I = I_0 \& \overline{E}])\\
	&= 2\Pr[\overline{E}]\\
	&\le 1/n^c
	\end{align*}
	so the total variation distance between the output distributions of $\ORL$ and $\ORL_1$ is at most $1/n^c$. Next, we show that $\ORL_1(G)$ produces the same distribution over sets as the following algorithm, $\SORL_0(G)$:
	
	\begin{algorithm}[H]
		\SetAlgoLined
		\DontPrintSemicolon
		\KwData{an unweighted graph $G$}
		\KwResult{an independent set $I$}
		
		$I\gets \emptyset$\;
		
		\For{$r=n,n-1,\hdots,1$}{
			
			$u_r\gets $ uniformly random element of the set $V\setminus \{u_n,u_{n-1},\hdots,u_{r+1}\}$\;
			
			\If{$u_r$ does not have a neighbor $v$ for which $v = u_s$ for some $s > r$}{
				
				Add $u_r$ to $I$\;
				
			}
			
		}
		
		\Return $I$\;
		
		\caption{$\SORL_0(G)$}
	\end{algorithm}
	Since the $u_r$s are selected without replacement from $V$, the distribution over tuples $(u_n,u_{n-1},\hdots,u_1)$ is a uniform distribution over permutations of $V$. Let $\mathcal R\subseteq \{1,2,\hdots,n\}^n$ and $\mathcal U\subseteq V^n$ denote the families of $\{1,2,\hdots,n\}$ and $V$-tuples with distinct coordinates respectively. Fix an ordering $v_1,v_2,\hdots,v_n$ of vertices in $G$ and let $\tau: \mathcal R\rightarrow \mathcal U$ be the map
	
	$$\tau(r_1,r_2,\hdots,r_{n-1},r_n) = (v_{r_1},v_{r_2},\hdots,v_{r_{n-1}},v_{r_n})$$
	$\tau$ is a bijection. Furthermore, for any tuple $\textbf{r} \in \mathcal R$, $\ORL_1(G)$ with $v_i$-rank $\textbf{r}_i$ outputs the same set $S$ as $\SORL_0(G)$ with vertex ordering $\textbf{u} = \tau(\textbf{r})$. Therefore, $\ORL_1(G)$ outputs the same distribution over sets as $\SORL_0(G)$.
	
	Finally, $\SORL_0(G)$ produces the same distribution over sets as $\SORL(G)$, because the permutation can be sampled with replacement and rejection of previous samples (as in $\SORL$) rather than without replacement (as in $\SORL_0$). Therefore, $\SORL(G)$ and  $\ORL(G)$ produce distributions over sets $S$ with total variation distance at most $1/n^c$, as desired.
\end{proof}

To lower bound the size of the independent set produced by $\SORL(G)$, we use an exposure martingale. We start by stopping the $\SORL$ algorithm early: we only consider the first $k = n/(2(\Delta+1))$ iterations. Each iteration samples 1 vertex, which precludes at most $\Delta$ other vertices from joining the independent set in the future. Therefore, after $k$ iterations, a randomly sampled vertex has a probability of at least $\frac{n - (\Delta+1) k}{n} \ge 1/2$ of still being able to join the independent set. Thus, $I$ has size at least $(\frac{1}{2})(\frac{n}{2(\Delta+1)}) = \frac{n}{8(\Delta+1)}$ in expectation. To obtain a high-probability lower bound on the size of $I$, we use the following proposition, which we prove using Azuma's Inequality:

\begin{proposition}\label{prop:nonshift}
	Consider a set $\mathcal{X}$, a distribution $\mathcal D$ over $\mathcal X$, a collection $X_1,X_2,\hdots,X_k$ of independent, identically distributed random variables sampled from $\mathcal D$, and a collection of functions $f_1,f_2,\hdots,f_k$, where $f_i: \mathcal{X}^i\rightarrow \mathbb{R}$. Suppose that there are numbers $M_0,M_1 > 0$ such that for all $i\in \{1,2,\hdots,k-1\}$ and all tuples $x_1,x_2,\hdots,x_i\in \mathcal X$, the following conditions hold:
	
	\begin{enumerate}
		\item(Max change) $|f_{i+1}(x_1,x_2,\hdots,x) - f_i(x_1,x_2,\hdots,x_i)|\le M_0$ and $|f_1(x)|\le M_0$ for all $x\in \mathcal X$
		
		\item(Expected increase) $\mathbb{E}_{X\sim \mathcal D} [f_{i+1}(x_1,x_2,\hdots,x_i,X)] \ge M_1 + f_i(x_1,x_2,\hdots,x_i)$ and $\mathbb{E}_{X\sim \mathcal D} [f_1(X)] \ge M_1$.
	\end{enumerate}
	Then
	
	$$\Pr[f_k(X_1,X_2,\hdots,X_k) < k M_1 - t] \le \exp\left(\frac{-t^2}{8M_0^2 k}\right)$$
\end{proposition}

\begin{proof}
	Set up a martingale $\{Y_i: i=0,1,\hdots,k\}$ based on the sequence of function values. Let $Y_0 = 0$ and for all $i \in \{1,2,\hdots,k\}$, let
	
	\begin{align*}
	Y_i &= f_i(X_1,X_2,\hdots,X_i) - \mathbb{E}[f_i(X_1,X_2,\hdots,X_i) | X_1,\hdots,X_{i-1}] + Y_{i-1}\\
	\end{align*}
	For all integers $i\ge 1$, $\mathbb{E}[Y_i | X_1,\hdots,X_{i-1}] = Y_{i-1}$, so the $Y_i$s are a martingale. Let $f_0$ denote the constant function $f_0 = 0$. By the \emph{Max change} condition,
	
	\begin{align*}
	|Y_i - Y_{i-1}| &= |f_i(X_1,X_2,\hdots,X_{i-1},X_i) - \mathbb{E}[f_i(X_1,X_2,\hdots,X_{i-1},X_i) | X_1,\hdots,X_{i-1}]|\\
	&\le \max_{a,b\in \mathcal X} |f_i(X_1,X_2,\hdots,X_{i-1},a) - f_i(X_1,X_2,\hdots,X_{i-1},b)|\\
	&\le \max_{a,b\in \mathcal X} (|f_i(X_1,X_2,\hdots,X_{i-1},a) - f_{i-1}(X_1,\hdots,X_{i-1})|\\
	&+ |f_{i-1}(X_1,\hdots,X_{i-1}) - f_i(X_1,X_2,\hdots,X_{i-1},b)|)\\
	&\le 2M_0\\
	\end{align*}
	Therefore, by Fact \ref{fact:Azuma},
	
	$$\Pr[Y_k < -t] \le \exp\left(-\frac{t^2}{8kM_0^2}\right)$$
	By the \emph{Expected increase} condition, for all integers $i\ge 1$,
	
	\begin{align*}
	Y_i &= f_i(X_1,X_2,\hdots,X_i) - \mathbb{E}[f_i(X_1,X_2,\hdots,X_i) | X_1,\hdots,X_{i-1}] + Y_{i-1}\\
	&= f_i(X_1,X_2,\hdots,X_i) - f_{i-1}(X_1,\hdots,X_{i-1})\\
	&+ f_{i-1}(X_1,\hdots,X_{i-1}) - \mathbb{E}[f_i(X_1,X_2,\hdots,X_i) | X_1,\hdots,X_{i-1}] + Y_{i-1}\\
	&\le f_i(X_1,X_2,\hdots,X_i) - f_{i-1}(X_1,\hdots,X_{i-1}) - M_1 + Y_{i-1}\\
	\end{align*}
	As a result,
	
	$$Y_k\le f_k(X_1,\hdots,X_k) - kM_1$$
	which means that
	
	$$\Pr[f_k(X_1,\hdots,X_k) < kM_1 - t] \le \Pr[Y_k < -t] \le \exp\left(-\frac{t^2}{8kM_0^2}\right)$$
	as desired.
\end{proof}

We now prove our main result by letting $f_i$ denote the size of the independent set after $i$ iterations of $\SORL$:

\begin{theorem}
	For any $c > 1$, there is an $(c)$-round CONGEST algorithm $\ORL(G)$ that, given a parameter $p \in (0,1)$ and an $n$-vertex graph $G$ with max degree $\Delta \le n/(256 \log(1/p)) - 1$, returns an independent set $I$ for which $|I| \ge n/(8(\Delta + 1))$ with probability at least $1 - p - 1/n^c$.
\end{theorem}

\begin{proof}
	$\ORL(G)$ is an $O(c)$-round algorithm in the CONGEST model. Furthermore, the set $I$ returned is an independent set because each vertex $v\in I$ has a strictly higher rank $r_v$ than its neighbors, which is not simultaneously possible for two adjacent vertices. Therefore, to prove the theorem, we just need to lower bound the size of the set $I$ returned by $\ORL(G)$. By Proposition \ref{prop:indep-equiv}, it suffices to show that the set $I$ returned by $\SORL(G)$ has size at least $n/(8(\Delta + 1))$ with probability at least $1 - p$.
	
	To lower bound the size of $I$, we apply Proposition \ref{prop:nonshift} with the following parameter settings:
	
	\begin{itemize}
		\item $k = n/(2(\Delta+1))$
		\item $\mathcal X = V(G)$
		\item $\mathcal D$: the uniform distribution over $\mathcal X$
		\item $X_i$: the vertex $u$ sampled during the $i$th iteration of the while loop in $\SORL(G)$.
		\item $f_i(x_1,x_2,\hdots,x_i)$: the function that maps a set of vertices $x_1,\hdots,x_i$ to $|I_i|$, where $I_i$ is the set $I$ between the $i$ and $(i+1)$th iterations of the while loop of $\SORL(G)$ with $u$ being $x_j$ in the $j$th while loop iteration.
		\item $M_0 = 1$
		\item $M_1 = 1/2$
		\item $t = k/4$
	\end{itemize}
	We now check the conditions of Proposition \ref{prop:nonshift} with each of these parameters. The \emph{Max change} condition follows immediately from the fact that $I_{i+1} = I_i$ or $I_{i+1} = I_i\cup \{X_{i+1}\}$ for all $i\ge 1$ and the fact that $|I_1|\le 1$, so we focus on the \emph{Expected increase} condition. Consider a set of choices $x_1,x_2,\hdots,x_i$ of the first $i$ while loop vertices $u$ and let $V_i = \{x_1,x_2,\hdots,x_i\}$ for all $i\in \{0,1,\hdots,k\}$. Let $X$ be a random variable denoting the $(i+1)$th vertex $u$ selected by the while loop from $U$. $X$ is uniformly chosen from $V$. By the if statement of the $\SORL$ algorithm, $X$ is added to $I$ if and only if $X$ is not equal to or adjacent to any vertex in $V_i$. There are at most $(\Delta + 1)|V_i|$ such vertices, so
	
	$$\Pr_{X\sim \mathcal D}[f_{i+1}(x_1,x_2,\hdots,x_i,X) \ne f_i(x_1,\hdots,x_i)] \ge 1 - \frac{(\Delta + 1)|V_i|}{n}$$
	Since $i\le k\le n/(2(\Delta + 1))$ and $|V_i| = i$, $\Pr_{X\sim \mathcal D}[f_{i+1}(x_1,x_2,\hdots,x_i,X) \ne f_i(x_1,\hdots,x_i)] \ge 1/2$ for all $i\in \{0,1,\hdots,k-1\}$. Furthermore,
	
	\begin{align*}
	\mathbb{E}_{X\sim \mathcal D} [f_{i+1}(x_1,x_2,\hdots,x_i,X)] &= \Pr[f_{i+1}(x_1,x_2,\hdots,x_i,X) \ne f_i(x_1,\hdots,x_i)] + f_i(x_1,x_2,\hdots,x_i)\\
	&\ge 1/2 + f_i(x_1,x_2,\hdots,x_i)
	\end{align*}
	
	Plugging in our lower bound on the probability shows that the \emph{Expected increase} condition is satisfied. Therefore, Proposition \ref{prop:nonshift} applies and shows that
	
	$$\Pr[|I_k| < k/4] \le \exp\left(-\frac{(k/4)^2}{8k}\right) = \exp(-k/128)$$
	In particular, since $|I| \ge |I_k|$,
	
	$$\Pr[|I| < n/(8(\Delta + 1))] \le \exp(-n/(256(\Delta + 1))) \le p$$
	Therefore, the independent set $I$ returned by $\SORL$ has the desired size with probability at least $1 - p$, as desired.
\end{proof}

\section{An Improved Approximation for Low-Arboricity Graphs}\label{sec:arb}

In this section we prove the following theorem. Recall that $\alpha$ is the arboricity of the input graph (see also Definition~\ref{def:arb}). Theorem~\ref{thm:mainArb} is a corollary of Theorems~\ref{thm:Wmain} and~\ref{thm:arb}.
\begin{theorem}\label{thm:arb}
	Let $\mathcal{A}$ be a $T$-round algorithm that finds a $(1+\epsilon)\Delta$-approximation for maximum-weight independent set in the CONGEST model. There is an $O(T\log n)$-round algorithm $\mathcal{A}'$ that finds a $8(1+\epsilon)\alpha$-approximation for maximum-weight independent set in the CONGEST model.
\end{theorem}

\paragraph{Description of algorithm $\mathcal{A}'$:} $\mathcal{A}'$ is formally given in Algorithm~\ref{alg:arb}.
The algorithm is similar to Algorithm~\ref{alg:amplify} that was presented in Section~\ref{sec:amplification}. The main difference is that we run $\mathcal{A}$ on subgraphs of maximum degree $4\alpha$, rather than on the input graph, which has maximum degree $\Delta$. We do this to ensure that the approximation factor we get is not too far from $4\alpha(1+\epsilon)$.
$\mathcal{A}'$ is formally given in Algorithm~\ref{alg:arb}. There are two stages, each consists of $\log n+1$ phases. 
\paragraph{First Stage:} Let $V^{4\alpha}_i$ be the set of nodes of degree at most $4\alpha$ in $G^i_{w_i}$, where $G^1_{w_1}=(V_1,E_1,w_1)=(V,E,w)$ is the original input graph, and for $i>1$, $G^i_{w_i}$ will be defined shortly. For $i=1$ to $\log n+1$ phases, in each phase, we run algorithm $\mathcal{A}$ on the subgraph induced by the nodes in $V^{4\alpha}_i$ and insert the resulting independent set $I_i$ into a stack $S$, that is initially empty. We then continue to the graph $G^{i+1}_{w_{i+1}}=(V_{i+1},E_{i+1},w_{i+1})$, that is defined as follows. Let $w_1=w$ be the original input weight function. We begin with defining $w_{i+1}$.

\begin{align*}
&w_{i+1}(v)=\begin{cases}
0 & \mbox{if } v\in V^{4\alpha}_i\\
w_i(v)-\sum_{u\in N(v)\cap I_i}w_i(u) & \mbox{otherwise}
\end{cases}
\end{align*}

Unlike the algorithm presented in Section~\ref{sec:amplification}, we reduce the total current weight, $w_i(v)$, from every $v\in V^{4\alpha}_i$, and not only from the nodes in $I_i$. For every node $v\notin V^{4\alpha}_i$, we reduce the total weight of the neighboring nodes in $I_i$. Similarly to Section~\ref{sec:amplification}, we define the following reduced weight function $w'_i(v)=w_i(v)-w_{i+1}(v)$, for any $i$. 

%\begin{align*}
%w'_i(v)=w_i(v)-w_{i+1}(v)
%\end{align*}

The idea behind the definition of $w_i$ and $w'_i$ is that $I_i$ is a good approximation, as to be shown later, to the graph $G^i_{w'_i}=(V_i,E_i,w'_i)$, which is the same graph as $G^i_{w_i}$ but with weight function $w'_i$ rather than $w_i$. Therefore, $I_1$ is a good approximation to the graph $G^1_{w'_1}$, which is the same graph as the input graph $G_{w}$, but with weight function $w'_1$ rather than $w$. Hence, intuitively, after finding $I_1$, we just need to find an independent set $I'$  that is a good approximation to the graph $(V,E,w-w'_1)$, and incorporate the two independent sets $I_1$ and $I'$ together to construct an independent set $I$ that is a good approximation to the original input graph $G=(V,E,w)$. The main issue with applying this approach recursively on $G_{w_2}=(V,E,w-w'_1)$ is that the set of nodes is the same as the set of nodes of $G_{w}$, and therefore, the set of nodes of maximum degree $4\alpha$ in $G_{w_2}$ is the same as in $G_{w}$, which implies that we don't make any progress by running $\mathcal{A}$ again on the same set of nodes. Therefore, in order to guarantee that our algorithm makes progress, we run $\mathcal{A}$ on $G^2_{w_2}=(V_2,E_2,w_2)=(V_2,E_2,w-w'_1)$, where we \emph{keep in $V_2$ only nodes of positive weight in $w_2$}. Since we can assume without loss of generality that $\mathcal{A}$ never picks nodes of non-positive weight, we can keep only the nodes of positive weight in $V_2$. Furthermore, by keeping the nodes of positive weight in $V_2$, we ensure that every $v\in V^{4\alpha}_1$ is not in $V_2$, and therefore our algorithm makes progress by running $\mathcal{A}$ again on the set of nodes of maximum degree $4\alpha$ in $G^2_{w_2}$, which is denoted by $V^{4\alpha}_2$, which is necessarily different than $V^{4\alpha}_1$. Hence, we define $V_{i+1}$ and $E_{i+1}$ as follows.

\begin{align*}
&V_{i+1}=\{v\in V_i\mid w_{i+1}(v)>0\}\\
&E_{i+1}=\{\{u,v\}\in E\mid u,v\in V_{i+1}\}\\
\end{align*}

This concludes the first stage.

\paragraph{Second Stage:} We construct an independent set $I$ as follows. For $i=0$ to $\log n$ phases, we pop out $I_{\log n+1-i}$ from the stack, and insert each $v\in I_{\log n+1-i}$ to $I$ unless $I$ already contains a neighbor of $v$. 

In Lemma~\ref{arb:approx}, we prove that $I$ is a $8(1+\epsilon)\alpha$-approximation for maximum-weight independent set in the original input graph $G_w$. We start with the following helper propositions.

\begin{proposition}\label{prop:arbRuntime}
	For any $i\geq 2$, it holds that $|V_{i+1}|\leq |V_i|/2$.
	\begin{proof}
		Since the arboricity of $G$ is $\alpha$, it holds that for any subgraph $H=(V',E')\subseteq G$ there are at least $|V'|/2$ nodes of degree at most $4\alpha$. Hence, for any $i\geq 1$, it holds that $|V_i^{4\alpha}|\geq |V_i|/2$. Since all the nodes in $V_i^{4\alpha}$ have zero weight in $w_{i+1}$, and $V_{i+1}=\{v\in V_i\mid w_{i+1}(v)>0\}$, the claim follows.
	\end{proof}
\end{proposition}

	\begin{proposition}\label{prop:arbWeight}
			$I_i$ is a $8\alpha(1+\epsilon)$-approximation for maximum-weight independent set in $G^i_{w'_i}$
		\end{proposition}
		\begin{proof}
			First, observe that, for every $v\in V^{4\alpha}_i$, $w'_i(v)=w_i(v)$, and for every $v\in V_i\setminus  V^{4\alpha}_i$, $w'_i(v)= \sum_{u\in N(v)\cap I_i} w_i(u)$. In total, we have that,
			\begin{align*}
			&w_i'(V_i)= w'_i(V^{4\alpha}_i)+ w_i'(V_i\setminus V^{4\alpha}_i)
			= w'_i(V^{4\alpha}_i)+\sum_{v\in V_i\setminus V^{4\alpha}_i} w'_i(N(v)\cap I_i)
			\leq w'_i(V^{4\alpha}_i)+ 4\alpha w'_i(I_i)
			\end{align*}
			where the last  inequality holds because each node in $I_i$ has at most $4\alpha$ neighbors in $V_i$, and $w_i(v)=w'_i(v)$, for any $v\in I_i$. 
			Let $G^{4\alpha}_{w'_i}=(V_i^{4\alpha},E(V_i^{4\alpha}),w'_i)$ be the subgraph of $G^i_{w'_i}$ induced by the nodes in $V^{4\alpha}_i$ with weight function $w'_i$. Let $OPT(G^{4\alpha}_{w'_i})$ be the value of an optimal solution in $G^{4\alpha}_{w'_i}$. Recall that $G^i_{w'_i}=(V_i,E_i,w'_i)$. Observe that since $w'_i(V_i)\leq w'_i(V^{4\alpha}_{i})+4\alpha  w'_i(I_i)$, 
			%\begin{align*}
			%w'_i(V_i)\leq w'_i(V^{2\alpha}_{i})+2\alpha  %w'_i(I_i) 
			%\end{align*}
			it follows that 
			\begin{align*}
			&OPT(G^i_{w'_i})\leq OPT(G^{4\alpha}_{w'_i})+4\alpha w'_i(I_i)\\
			&\Rightarrow \frac{OPT(G^i_{w'_i})-OPT(G^{4\alpha}_{w'_i})}{4\alpha}\leq w'_i(I_i)
			\end{align*}
			
			Finally, since $\mathcal{A}$ is a $(1+\epsilon)\Delta$-approximation algorithm, and since the maximum degree of $G^{4\alpha}_{w'_i}=G^{4\alpha}_{w_i}$ is at most $4\alpha$, it holds that $I_i$ is a $4(1+\epsilon)\alpha$-approximation for $G^{4\alpha}_{w'_i}$. Hence,
			\begin{align*}
			 &w'_i(I_i)\geq \frac{OPT(G^i_{w'_i})-OPT(G^{4\alpha}_{w'_i})}{4\alpha}\\
			& \Rightarrow 2 w'_i(I_i)\geq \frac{OPT(G^i_{w'_i})-OPT(G^{4\alpha}_{w'_i})+OPT(G^{4\alpha}_{w'_i})}{4\alpha(1+\epsilon)}\\
			&\Rightarrow w'_i(I_i)\geq \frac{OPT(G^i_{w'_i})}{8(1+\epsilon)\alpha}
			\end{align*}
			As desired. 
		\end{proof}
\begin{lemma}\label{arb:approx}
	$I$ is a $8(1+\epsilon)\alpha$-approximation to $OPT(G_w)$.
	\begin{proof}
		First, by Proposition~\ref{prop:arbRuntime}, $V_{\log n+2}=\emptyset$. Therefore, for every $v\in V$, there is a phase $i_v\in [\log n+1]$ for which $v\in \bigcup_{i=1}^{i_v}V_i$ and $v\notin V_{i_v+1}$. Hence,
		\begin{align*}
		w(v)=\sum_{i=1}^{i_v}w'_{i}(v)=\sum_{i\in [\log n+1]: i\in V_i} w'_i(v)
		\end{align*}
		In total, we have that, $w(V)=\sum_{v\in V}\sum_{i\in [\log n+1]: i\in V_i} w'_i(v)=\sum_{i=1}^{\log n+1} w'_i(V_i)$.
		%\begin{align*}
		%w(V)=\sum_{v\in V}\sum_{i\in [\log n+1]: i\in %V_i} w'_i(v)=\sum_{i=1}^{\log n+1} w'_i(V_i)
		%\end{align*}
		Observe that the second stage in Algorithm~\ref{alg:arb} is identical to the second stage of Algorithm~\ref{alg:amplify}. Therefore, we can apply Proposition~\ref{prop:stack} from Section~\ref{sec:amplification} to argue that $w(I)\geq \sum_{i=1}^{\log n+1}w'_i(I_i)$.
		%\begin{align*}
		%w(I)\geq \sum_{i=1}^{\log n+1}w'_i(I_i)
		%\end{align*}
		Hence, 
		\begin{align*}
		&w(I)\geq\sum_{i=1}^{\log n+1}w'_i(I_i)
		\geq \sum_{i=1}^{\log n+1}\frac{OPT(G^i_{w'_i})}{8(1+\epsilon)\alpha}
		\geq \frac{OPT(G_w)}{8(1+\epsilon)\alpha}
		\end{align*}
		as desired, where the second inequality follows by Proposition~\ref{prop:arbWeight}.
		\end{proof}
\end{lemma}
	\begin{algorithm}[H]\label{alg:arb}
	\SetAlgoLined
	\DontPrintSemicolon
	\KwData{a graph $G=(V,E,w)$}
	\KwResult{an independent set $I$}
	
	$I\gets \emptyset$\\
	$S\gets \emptyset$\tcp*{$S$ is the Stack}
	$V_1=V$\\
	$E_1=E$\\
	$w_1=w$\\

	\For{$i=1$ to $\log n+1$ phases}{
		Let $V^{4\alpha}_i=\{v\in V_i\mid deg(v)\leq 4\alpha\}$\\
		$E^{4\alpha}_i=\{\{u,v\}\in E\mid u,v\in V^{4\alpha}_i\}$\\
		Run $\mathcal{A}$ on $G^{4\alpha}_{w_i}=(V^{4\alpha}_i,E^{4\alpha}_i,w_i)$\\
		Let $I_i\gets \mathcal{A}(G^{4\alpha}_{w_i})$\\
		Insert all the nodes in $I_i$ into $S$\\
		$\forall v\in V_i:$\\
		\begin{align*}
		&w_{i+1}(v)=\begin{cases}
		0 & \mbox{if } v\in V^{4\alpha}_i\\
		w_i(v)-\sum_{u\in N(v)\cap I_i}w_i(u) & \mbox{otherwise}
		\end{cases}
		\end{align*}
		
		$V_{i+1}=\{u\in V_i\mid w_{i+1}(u)> 0\}$\\
		$E_{i+1}=\{\{u,v\}\in E\mid u,v\in V_{i+1}\}$\\
	}
	\For{$i=0$ to $\log n$ rounds}{
		
		Pop $I_{\log n+1-i}$ from $S$\\
		\For{$v\in I_{\log n+1-i}$}{  \If{$N(v)\cap I=\emptyset$}
			{add $v$ to $I$}
		}
	}
	\Return $I$\\
	\caption{$8(1+\epsilon)\alpha$-approx}
\end{algorithm}
\section{Lower Bound}\label{sec:LB}

In this section, we prove Theorem \ref{thm:LB}. We do this by a reduction to Naor's \cite{Naor91} lower bound (Theorem \ref{fact:mis-lb}). Our reduction is described by the following lemma:

\begin{lemma}\label{lem:reduction}
	Suppose that there exists a $T(n)$-round algorithm $\AIS(G)$ in the LOCAL model that outputs an independent set containing at least $n/(c\Delta)$ vertices in an $n$-vertex graph $G$ with probability at least $1 - p(n)$, where $p$ is a decreasing function. Then, for any integer $n_1$, there is an $O(cT(n_0n_1))$-round algorithm $\LMIS(C)$ in the LOCAL model that outputs a maximal independent set of an $n_0$-vertex cycle graph $C$ with probability at least $1 - n_0 p(n_1)$ as long as $n_1\ge n_0$.
\end{lemma}

Let $C$ be a cycle graph of $n_0$ nodes. Algorithm $\LMIS(C)$ runs $\mathcal{A}$ on a graph $C_1$ which is a cycle of cliques, and will be formally defined shortly. After $\mathcal{A}$ finds an independent set of size $n/c\Delta$ in $C_1$, it uses the resulting independent set to find a maximal independent set in $C$. We now formally define the graph $C_1$. For an $n_0$-vertex cycle $C$ consisting of vertices $u_1,u_2,\hdots,u_{n_0}$ in that order, let $C_1$ be a graph on $n_0n_1$ vertices $\{\{v_{ij}\}_{j=1}^{n_1}\}_{i=1}^{n_0}$. There is an edge between two vertices $v_{ij},v_{i'j'}$ in $C_1$ if and only if $|i' - i|\le 1$ or $i' = n_0$ and $i = 1$. The ID of a vertex $v_{ij}$ in $C_1$ is the concatenation of the ID for $u_i$ in $C$ and the number $j$. Notice that these IDs have length at most $\log (n_0n_1)$. $C_1$ is a cycle of cliques, with a biclique between two adjacent cliques. This graph is depicted in Figure 1.

\begin{figure}\label{fig:clique-cycle}
\begin{center}
\includegraphics[width=\linewidth,height=18cm]{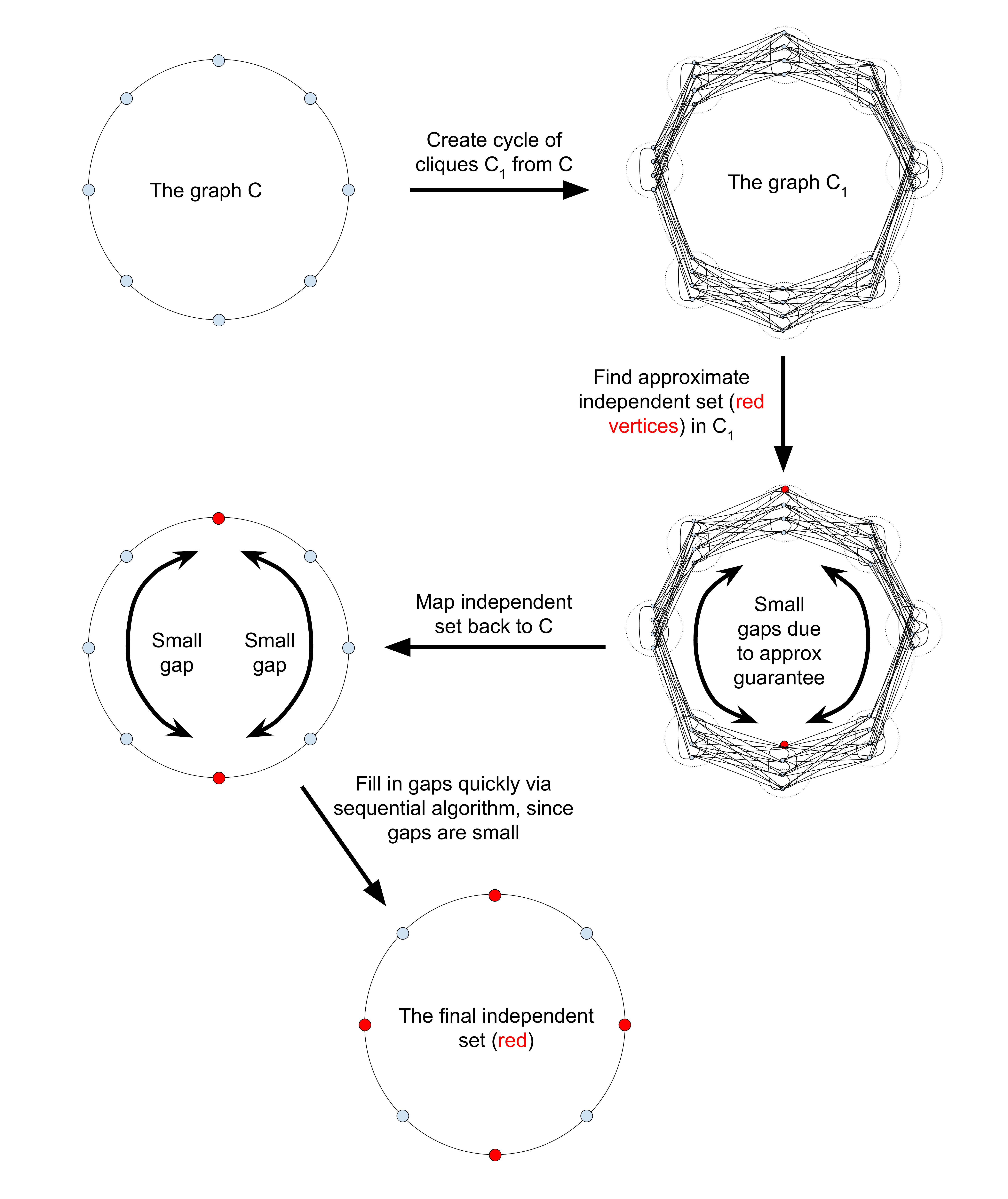}
\end{center}
\caption{An illustration for our reduction to prove the lower bound. To find a maximal independent set in a cycle $C$, the nodes run an approximate-$\MaxIS$ algorithm to find an independent set $I_1$ in $C_1$, which is obtained from $C$ as follows. Each node $v\in C$ is replaced by a large clique of size $\approx 2^{|C|}$, denoted by $D(v)$, where every two adjacent cliques are connected by a bi-clique. Using the independent set $I_1$ in $C_1$, the nodes map it to find an independent set $I$ in $C$, as follows. Every $v\in C$ joins $I$ if and only if $D(v)$ contains a node in $I_1$. Due to the approximation guarantee, the gap between any two nodes in $I_1$ in $C_1$ is small, and therefore the gap between any two nodes in $I$ in $C$ is also small. Finally, the nodes run a greedy $\MIS$ algorithm to ``fill in" the gaps, and find an $\MIS$ in $C$.}
\end{figure}

$\LMIS(C)$ starts by computing an independent set $I_1$ using $\AIS(C_1)$. This randomized LOCAL algorithm can be implemented in the LOCAL model on $C$, with each vertex $u_i$ simulating all of the actions performed by $\AIS(C_1)$ on the vertices $\{v_{ij}\}_{j=1}^{n_1}$. The set $I_1$ maps to an independent set $I$ in $C$ since each $u_i$ maps to a clique in $C_1$. Because the $T(n_0n_1)$-neighborhood of each vertex in $C_1$ is a $\ge n_1$-vertex graph and the algorithm $\AIS$ is distributed, a $O(T(n_0n_1))$-neighborhood of a vertex in $C$ must contain a vertex in $I$ with probability at least $1 - p(n_1)$. By a union bound over all cliques in $C_1$, the distance between any two consecutive vertices in $I$ along the cycle is at most $O(T(n_0n_1))$ with probability at least $1 - n_0 p(n_1)$. Therefore, all connected components of $C\setminus I$ have size at most $O(T(n_0n_1))$, so sequentially finding a maximal independent set in each component simultaneously takes $O(T(n_0n_1))$ time to extend $I$ to an MIS for $C$.

In the proof of Lemma \ref{lem:reduction}, we crucially exploit two properties of the algorithm $\AIS$ that follow from its correctness:

\begin{enumerate}
	\item $\AIS$ is \emph{globally consistent} in the sense that $\AIS(G)$ returns an \emph{independent set} of $G$ with high probability (see Proposition \ref{prop:max-correct}).
	\item $\AIS$ is \emph{locally present} in the sense that a $O(cT(n_0n_1))$-neighborhood of any vertex intersects $\AIS(C_1)$ with high probability (see Proposition \ref{prop:ngbrhood}).
\end{enumerate}

Notice that the second property does not hold for one round of Boppanna's algorithm and that the first property does not hold for a $o(\log^* n)$-time greedy algorithm.

We now implement the algorithm $\LMIS$, which (1) calls $\AIS$ on $C_1$, (2) maps the found independent set back to $C$ and (3) finds a maximal independent set in the connected components between consecutive independent set vertices:

\begin{algorithm}[H]
	\SetAlgoLined
	\DontPrintSemicolon
	\KwData{an $n_0$-vertex cycle graph $C$}
	\KwResult{an maximal independent set $S$ of $C$}
	
	$S\gets \emptyset$\;
	
	$I_1\gets \AIS(C_1)$ (implemented in LOCAL model on $C$ as stated in Proposition \ref{prop:mis-runtime} proof) \tcp*{step (1)}\;
	
	$I\gets \{u_i\in V(C) \text{ for which there exists $j$ with } v_{ij}\in I_1\}$\; \tcp*{step (2)}
	
	Add $I$ to $S$\;
	
	$J\gets \{u\in V(C): u\in I \text{ or $u$ is adjacent to a vertex in $I$}\}$\;
	
	$C_2\gets C\setminus J$\;
	
	\For{each connected component $D$ of $C_2$ in parallel}{
		
		Add a fixed maximal independent set of $D$ to $S$\; \tcp*{step (3)}
		
	}
	
	\Return $S$\;
	
	\caption{$\LMIS(C)$}
\end{algorithm}

We start by showing the correctness of the algorithm.

\begin{lemma}\label{prop:max-correct}
	$\LMIS(C)$ outputs a maximal independent set of the $n_0$-vertex cycle graph $C$ with probability at least $1 - p(n_1)$.
\end{lemma}

\begin{proof}
	By definition of the algorithm $\AIS$, $I_1$ is an independent set of $C_1$ with probability at least $1 - p(n_0n_1)\ge 1 - p(n_1)$. For the rest of the proof, assume that $I_1$ is an independent set $C_1$ (the complement happens with probability at most $p(n_1)$). We now show that $I$ is an independent set. Suppose, instead, that there exist adjacent $u_i,u_{i+1}\in I$. By definition of $I$, there exist vertices $v_{ij},v_{(i+1)j'}\in I_1$. By construction of $C_1$, $v_{ij}$ and $v_{(i+1)j'}$ are adjacent vertices in $C_1$, a contradiction to the fact that $I_1$ is an independent set in $C_1$. Therefore, $I$ must be an independent set in $C$.
	
	No vertices in $C_2$ are adjacent to vertices in $I$ within $C$ by definition of $C_2$. Therefore, $S$ is an independent set in $C$. Furthermore, each vertex in $J$ is adjacent to a vertex in $I$, while each vertex in $V(C)\setminus J$ is adjacent to a vertex in the maximal independent set computed for $C_2$. Therefore, $S$ is also maximal at the end of the algorithm. Therefore, $S$ is a maximal independent set if $I$ is an independent set, which happens with probability at least $1 - p(n_1)$, as desired.
\end{proof}

The rest of the analysis focuses on the runtime. To bound the runtime, we need to exploit the fact that $\AIS$ is a distributed algorithm to show that the independent set returned has small gaps with high probability. This is shown by using the fact that $\AIS$, in the neighborhood of a vertex $v$, cannot distinguish between $C_1$ and an $O(T(n_0n_1))$-length cycle of cliques containing $v$. We formalize this in Proposition \ref{prop:ngbrhood-1}. Let $R_{\text{large}} = (100c + 1)T(n_0n_1) + 2$ and $R_{\text{small}} = 100cT(n_0n_1)$. Let $L_v$ denote the set of vertices $u\in V(C_1)$ for which the distance from $u$ to $v$ is at most $R_{\text{large}}$. Let $S_v$ denote the set of vertices $u\in V(C_1)$ for which the distance from $u$ to $v$ is at most $R_{\text{small}}$. Let $C_v$ denote the induced subgraph of $C_1$ with respect to the set of vertices $L_v$. Recall that $\AIS$ is a randomized algorithm that outputs a distribution over sets of vertices in the input graph. We now show the following property of this distribution:

\begin{proposition}\label{prop:ngbrhood-1}
	For any vertex $v\in V(C_1)$, $S_v\cap \AIS(C_1)$ has the same distribution as $S_v\cap \AIS(C_v)$.
\end{proposition}

\begin{proof}
	For any vertex $u\in V(C_1)$, let $f_u(C_1) = \textbf{1}_{u\in \AIS(C_1)}$; that is, the indicator function of $u$'s presence in the independent set $\AIS(C_1)$. Let $U_u$ be the set of vertices in $C_1$ with distance at most $T(n_0n_1)$ from $u$. Since $\AIS$ is a $T(n_0n_1)$-round algorithm in the LOCAL model, $f_u$ is only a function of the randomness, IDs, and edges incident with vertices in $U_u$ for any $u\in V(C_1)$. By definition of $f_u$,
	
	$$S_v\cap \AIS(C_1) = \{u\in S_v: f_u(C_1) = 1\}$$
	The set $S_v\cap \AIS(C_1)$ is therefore only a function of the randomness, IDs, and edges incident with vertices in $\cup_{u\in S_v} U_u$. All of this information is contained in the graph $C_v$, since any vertex in the set $\cup_{u\in S_v} U_u$ is within distance $R_{\text{small}} + T(n_0n_1) = R_{\text{large}} - 2$ of $v$. Therefore, $S_v\cap \AIS(C_1)$ is only a function of randomness, IDs, and edges in the graph $C_v$, which means that $S_v\cap \AIS(C_1)$ is identically distributed to $S_v\cap \AIS(C_v)$, as desired.
\end{proof}

As a result, to show that $S_v$ contains a vertex of the independent set with high enough probability, it suffices to think about $C_v$ instead of $C_1$. In this proposition, we exploit the fact that $\AIS$ returns an independent set with size at least $n/(c\Delta)$ on $n$-vertex graphs with maximum degree higher than $\Omega(n/\log^* n)$:

\begin{proposition}\label{prop:ngbrhood-2}
	For any vertex $v\in V(C_1)$, $S_v\cap \AIS(C_1)\ne \emptyset$ with probability at least $1 - p(n_1)$.
\end{proposition}

\begin{proof}
	By Proposition \ref{prop:ngbrhood-1}, $S_v\cap \AIS(C_v)$ is identically distributed to $S_v\cap \AIS(C_1)$, so it suffices to lower bound the probability that $S_v\cap \AIS(C_v)$ is empty. Let $I_v := \AIS(C_v)$. The maximum degree of vertices in $C_v$ is $3n_1$. Furthermore, $|V(C_v)| = (2R_{\text{large}} + 1)n_1 \ge 200cT(n_0n_1)n_1$. By the output guarantee of $\AIS$, $I_v$ is an independent set and $|I_v|\ge |V(C_v)|/(c (3n_1))\ge 60T(n_0n_1)$ with probability at least $1 - p(n_1)$. The vertices on $V(C_v)\setminus S_v$ are a union of $2(R_{\text{large}} - R_{\text{small}}) \le 4T(n_0n_1)$ cliques on $n_1$ vertices. Therefore, since $I_v$ is an independent set, $|I_v\cap (V(C_v)\setminus S_v)|\le 4T(n_0n_1)$. Therefore, $|I_v\cap S_v| \ge 60T(n_0n_1) - 4T(n_0n_1) > 0$ with probability at least $1 - p(n_1)$, as desired.
\end{proof}

Now, we union bound to prove the desired property for all intervals with width $2R_{\text{small}}$:

\begin{proposition}\label{prop:ngbrhood}
	Let $I$ be the output of $\AIS(C_1)$. With probability at least $1 - n_0 p(n_1)$, $S_v\cap I\ne\emptyset$ for all $v\in V(C_1)$.
\end{proposition}

\begin{proof}
	By Proposition \ref{prop:ngbrhood-2} and a union bound over all $i\in \{1,2,\hdots,n_0\}$, $S_{v_{i1}}\cap I\ne\emptyset$ for all $i\in \{1,2,\hdots,n_0\}$ with probability at least $1 - n_0 p(n_1)$. For any $j\in \{1,2,\hdots,n_1\}$, $S_{v_{ij}} = S_{v_{i1}}$. Since every vertex in $C_1$ is equal to $v_{ij}$ for some $i\in \{1,2,\hdots,n_0\}$ and $j\in \{1,2,\hdots,n_1\}$, $S_v\cap I\ne\emptyset$ for all $v\in V(C_1)$ with probability at least $1 - n_0 p(n_1)$, as desired.
\end{proof}

We now prove a runtime bound:

\begin{proposition}\label{prop:mis-runtime}
	Given an $n_0$-vertex cycle graph $C$, $\LMIS(C)$ runs in $O(T(n_0n_1))$ time with probability at least $1 - n_0 p(n_1)$.
\end{proposition}

\begin{proof}
	We go through the $\LMIS$ algorithm line by line. The call to $\AIS(C_1)$ can be implemented in the LOCAL model on $C$ as follows. Any $T$-round LOCAL algorithm can be viewed as independently flipping coins at each vertex and sending the IDs and randomness of a vertex $u$ to each vertex $v$ in its $T$-neighborhood, followed by no additional communication. This communication can be done in $C$ by having $u_i$ generate the randomness used by all $v_{ij}$s in $\AIS(C_1)$. Then, $u_i$ sends this randomness and the IDs of all $v_{ij}$s to each vertex in the $T(n_0n_1)$-neighborhood of $u_i$ in $C$. Finally, the $\AIS$ algorithm's execution on $v_{ij}$ can be run on $u_i$ instead. Thus, the call to $\AIS(C_1)$ takes at most $T(n_0n_1)$ rounds.
	
	$I$, $J$, and $C_2$ can each be computed in at most two rounds. By Proposition \ref{prop:ngbrhood}, $C_2$ has connected components with size at most $O(T(n_0n_1))$ with probability at least $1 - n_0 p(n_1)$. Thus, for each connected component $D$ of $C_2$, the vertices $u\in D$ can be sent $D$ in $O(T(n_0n_1))$ rounds. With no futher communication, the vertices $u\in D$ each use the same algorithm to compute a maximal independent set of $D$. This completes all lines of the algorithm. Therefore, the algorithm takes $O(T(n_0n_1))$ time with probability at least $1 - n_0 p(n_1)$, as desired.
\end{proof}

\begin{proof}[Proof of Lemma \ref{lem:reduction}]
	Follows immediately from Propositions \ref{prop:max-correct} ($S$ is an MIS) and \ref{prop:mis-runtime} (for runtime).
\end{proof}

Given Lemma \ref{lem:reduction}, we can now prove that any algorithm for approximate independent set that succeeds with arbitrarily high probability must take $\Omega(\log^* n)$ rounds:

\begin{theorem}\label{thm:full-lb}
	For any constant $b$, any randomized $o(\log^* n)$-time algorithm that computes an independent set with size greater than $\Omega(n/\Delta)$ in an $n$-vertex graph succeeds with probability at most $1 - 1/(10\log^{(b)} n)$, where $\log^{(b)}(x)$ is the function defined recursively as $\log^{(0)}(x) = x$ and $\log^{(b)}(x) = \log \log^{(b-1)}(x)$. 
\end{theorem}

\begin{proof}
	Suppose, for the sake of contradiction, that there exists a LOCAL algorithm $\AIS(G)$ that, when given an $n$-vertex graph $G$, takes $o(\log^* n)$ time and outputs an $\Omega(n/\Delta)$-vertex independent set of $G$ with probability at least $1 - 1/(10\log^{(b)} n)$. For some value of $n_0$, define $n_1$ as follows. Let $n_1^{(0)} = n_0$, let $n_1^{(i)} = 2^{n_1^{(i-1)}}$ for all $i > 0$, and let $n_1 = n_1^{(b)}$. By Lemma \ref{lem:reduction}
	and the fact that $n_1\ge n_0$, there is an $o(\log^* (n_0n_1)) = o(b + \log^* n_0) = o(\log^* n_0)$-round LOCAL algorithm $\LMIS(C)$ that, given an $n_0$-vertex cycle graph $C$, outputs a maximal independent set of $C$ with probability at least $1 - n_0 (1/(10\log^{(b)}(n_1))) = 9/10$. The existence of such an algorithm contradicts Theorem \ref{fact:mis-lb}, as desired.
\end{proof}

\begin{proof}[Proof of Theorem \ref{thm:LB}]
Theorem \ref{thm:LB} is implied immediately by Theorem \ref{thm:full-lb} with $b = 2$.
\end{proof}

\section{Discussion}\label{sec:dis}
There are many interesting open questions that are left unsolved following this work. The first obvious open question is to close the gap between our upper and lower bounds. Here, we provide some more open questions. First, we showed that in the randomized case, when we usually look for algorithms that succeed with high probability, finding an $O(\Delta)$-approximation to MaxIS is strictly easier than MIS. An immediate interesting open question in whether the same holds for the deterministic case. The following is open for both the LOCAL and CONGEST models.

\begin{oq}
	Let $T_{\alpha}(\MaxIS)$ be the running time of finding an $\alpha$-approximation to $\MaxIS$ deterministically in the CONGEST model, and let $T(MIS)$ be the running time of finding a maximal independent set in the CONGEST model. Prove or disprove that there is a constant $c$ for which it holds that $$T_{c\Delta}(\MaxIS)=o(T(\MIS))$$
\end{oq}

Another interesting implication of our result is the following one. Observe that in the sequential setting, one can find a $(\Delta+1)$-approximation for $\MaxIS$ by finding a $(\Delta+1)$-vertex-colouring. This is because we can simply take the colour class of maximum weight, which is a $(\Delta+1)$-approximation for $\MaxIS$. However, in the distributed setting, it not clear how to use a colouring to find a good approximation for $\MaxIS$. This is because finding the colour class of maximum weight requires $\Omega(D)$ rounds, where $D$ is the diameter of the network, which can be as large as $n-1$. 

Interestingly, our upper and lower bounds for $O(\Delta)$-approximation for $\MaxIS$ match the best currently known upper and lower bounds for distributed $(\Delta+1)$-colouring~\cite{ChangLP18,linial1992locality}. While this doesn't necessarily imply any connection between the two problems in the distributed setting, it might hint for a possible one. The following is open for both the LOCAL and CONGEST models.

\begin{oq}
	Prove or disprove: Given a $T$ rounds algorithm for $(\Delta+1)$-colouring, it is possible to find an $O(\Delta)$-approximation to $\MaxIS$ in $O(T)$ rounds. 
\end{oq}

\paragraph{Acknowledgments:} We would like to thank Amir Abboud, Shafi Goldwasser, Siqi Liu, Sidhanth Mohanty, Omri Shmueli, Rotem Tsabary and Richard Zhang for fruitful discussions.

\bibliographystyle{plain}
\bibliography{paper}

\appendix

\section{Deferred Proofs from Section~\ref{sec:warmup}}\label{app:warm}

\begin{proof}[\textbf{Proof of Lemma~\ref{lem:MISSIZE}}]
	The idea is to pad $H$ with more vertices and then to run an algorithm for maximal independent set on the new graph. In fact, the easiest way to see this is to argue that $\mathcal{A}$ finds a maximal independent set with high probability on the graph $H'$ obtained by adding $n-n_H$ isolated nodes to $H$ with unique identifiers. Since any maximal independent set in $H'$ induces a maximal independent set in $H$, the claim follows. However, some of the algorithms in the CONGEST model assume that the input graph is connected\footnote{This assumption is usually made for global problems such as computing the diameter or all-pairs-shortest-paths. This is because global problems admit an $\Omega(D)$ lower bound, where $D$ is the diameter of the network, which is $\infty$ for disconnected graphs. While assuming connectivity might not seem reasonable for the MIS problem, for completeness, we want our reduction to hold even for algorithms that make this assumption.}. To get around the connectivity issue, we define the graph $H'$ obtained by adding a path of $\poly(n)$ nodes with unique $\Theta(\log n)$-bit identifiers to each node that is \emph{local minimum} in $H$ (with respect to the identifiers). Each node that is added to a path connected to a local minimum $u\in V_H$, is given a unique identifier starting with the $c\log n$ bits of the identifier of  $u$ as the LSB's (least significant bits), followed by another $c\log n$ bits to ensure that the identifier is unique with respect to the other nodes on the same path. Observe that $H'$ is a graph of $\poly(n)$ nodes, with unique identifiers of $\Theta(\log n)$ bits. Hence, $H'$ is an appropriate input to the CONGEST model. Furthermore, given a maximal independent set $I'$ of $H'$, one can easily find a maximal independent set in $H$, as follows. Let $I=I'\cap V_H$. Each node that is a local minimum in $H$ joins $I$ if none of its neighbors in $H$ is in $I$. It holds that $I$ (after adding the additional nodes) is a maximal independent set in $H$. Since the nodes in $H$ can easily simulate a maximal independent set algorithm in $H'$, without any additional communication cost, it follows that the total running time is $\MIS(|V_{H'}|,\Delta_{H'})+1$, where $V_{H'}$ and $\Delta_{H'}$ are the set of nodes and maximum degree in $H'$, respectively. Since $\Delta_{H'}\leq \Delta_H+1$, and $|V_{H'}|=\poly(n)$, it holds that $\MIS(|V_{H'}|,\Delta_{H'})=\MIS(\poly(n),\Delta_H)$. Moreover, for any $n$ we know that for the specific problem of finding a maximal independent set it holds that $\MIS(\poly(n),\Delta)=O(\MIS(n,\Delta))$. This is because the round-complexity of finding a maximal independent set is at most logarithmic in the number of nodes. Finally, since a maximal independent set algorithm in $H'$ succeeds with probability $1-1/\poly(|V_H'|)\geq 1-1/\poly(n)$, the claim follows.
\end{proof}

\end{document}